\documentclass[12pt]{article}

\usepackage{amsmath,amssymb,amsthm}
\usepackage{graphicx}
\usepackage[round]{natbib}
\usepackage[colorlinks=true,
            linkcolor=blue,
            citecolor=blue,
            urlcolor=blue]{hyperref}
\usepackage{url}
\usepackage{booktabs}
\usepackage[table]{xcolor}
\usepackage{hyperref}
\usepackage{dsfont}
\usepackage{bm}
\usepackage{placeins}
\usepackage{multirow}
\usepackage{subcaption}
\usepackage{adjustbox}

\usepackage{algorithm}
\usepackage{algcompatible}
\usepackage[noend]{algpseudocode}

\makeatletter
\def\BState{\State\hskip-\ALG@thistlm}
\makeatother

\newtheorem{theorem}{Theorem}

\newtheorem{proposition}{Proposition}
\newtheorem{lemma}{Lemma}

\usepackage{lineno}
\usepackage{setspace}
\usepackage[margin=0.7in]{geometry}
\newcommand{\pdl}[1]{\textcolor{blue}{\textsf{#1}}}

\title{Bayesian Univariate-Guided Sparse Regression for Ultra-High-Dimensional Data}

\author{
Priyam Das\thanks{Corresponding author. 
Email: dasp4@vcu.edu. 
Mailing address: Department of Biostatistics, Virginia Commonwealth University,
830 E. Main Street, Richmond, VA 23219, USA.} \\
Department of Biostatistics, Virginia Commonwealth University
}

\date{}

\begin{document}

\maketitle
\begin{abstract}
Ultra-high-dimensional genomic studies require simultaneous identification of sparse signals, control of spurious discoveries, and computationally feasible inference when the number of predictors can reach hundreds of thousands or more. We propose Bayesian
Univariate-Guided Sparse Regression (BUGS), a global--local shrinkage framework that incorporates marginal association information directly into the prior through continuous modulation of shrinkage. By embedding univariate guidance within the nonlinear variance
structure of a regularized horseshoe prior, BUGS adaptively differentiates predictors with strong and weak marginal evidence without relying on hard screening. We establish theoretical results on prior concentration, posterior contraction, and guidance-induced shrinkage separation, and characterize the behavior of the prior under uninformative guidance. For ultra-high-dimensional problems, we develop BUGS-Active, an active-set MCMC approximation that reduces the cost of local-scale updates from $O(p)$ to $O(|A_n|)$ while retaining global coefficient updates. Simulations under independent and correlated designs show strong signal recovery and substantially lower empirical false discovery rates relative
to competing methods. We apply BUGS-Active to a DNA methylation study of developmental age with $n=1051$ samples and approximately $850{,}000$ CpG sites, jointly analyzing all probes without pre-screening. The analysis yields strong out-of-sample prediction and a
sparse set of age-associated CpG sites with stable posterior support. These results illustrate the practical value of marginally guided Bayesian shrinkage for large-scale genomic regression.
\end{abstract}

\noindent%

\section{Introduction}

Modern genomic, epigenomic, and microbiome studies routinely generate predictor spaces in
which the number of measured features greatly exceeds the number of available subjects
\citep{li_microbiome, Shokralla2012, Caporaso2011, BuhlmannKalischMeier2014}.
The disparity can be extreme in contemporary molecular assays. For example, the Illumina
MethylationEPIC platform measures more than $850{,}000$ CpG sites per individual
\citep{Pidsley2016}, while cohort sizes are commonly orders of magnitude smaller.
Such data create several intertwined statistical challenges: identifying a sparse and
interpretable set of associated features, limiting spurious discoveries among an enormous
number of candidates, accommodating substantial correlation among molecular measurements,
and retaining reliable prediction and uncertainty quantification. These challenges become
particularly acute when the scientific goal is not only prediction but also discovery of a
small set of features that can be interpreted and investigated further.

Penalized regression provides a computationally attractive approach to sparse estimation in
high dimensions. The Lasso and related procedures, including the elastic net and adaptive
Lasso, have become standard tools because they can be fitted efficiently even when $p\gg n$
\citep{Tibshirani1996, Zou2005, Zou2006, friedman2010glmnet}. When the ambient
dimension becomes very large, a complementary strategy is to first reduce the feature space
using inexpensive marginal information. Sure independence screening (SIS), for example,
ranks predictors by marginal association and can reduce an ultrahigh-dimensional problem
to a more manageable scale under suitable conditions \citep{fan2008sis}. Iterative
extensions were subsequently developed to recover predictors that may be weak marginally
but important jointly \citep{FanSamworthWu2009}, while related screen-and-clean procedures
separate preliminary dimension reduction from subsequent multivariable selection
\citep{WassermanRoeder2009}. More recently, methods such as UniLasso use univariate
regressions to guide penalized multivariable estimation \citep{chatterjee2025unilasso}.
Together, this literature demonstrates that marginal association can contain useful
information about variable relevance. At the same time, hard screening or explicitly
two-stage procedures may make inference sensitive to the initial reduction step, motivating
approaches that use marginal information without permanently excluding predictors from the
joint model.

Bayesian sparse regression offers a complementary route by regularizing coefficients through
prior distributions while propagating posterior uncertainty. Approaches include the Bayesian
Lasso and spike-and-slab formulations, as well as continuous global--local shrinkage priors
\citep{park2008blasso, rockova2018ssl, PolsonScott2011}. The latter have become
particularly useful for sparse high-dimensional inference because a common global scale can
strongly regularize the overall model while predictor-specific local scales allow sufficiently
strong signals to escape shrinkage. Important examples include the horseshoe
\citep{carvalho2010horseshoe}, horseshoe+ \citep{bhadra2017horseshoeplus},
Dirichlet--Laplace \citep{bhattacharya2015dl}, and R2D2 \citep{zhang2022r2d2} priors.
The regularized horseshoe further stabilizes large coefficients through a slab component
\citep{PiironenVehtari2017}, and theoretical work has established favorable, in some cases
near-optimal, posterior contraction behavior for global--local shrinkage priors under
sparsity \citep{vanderPasKleijnvdV2014, vanderPasSzabovdV2017}.

Despite these developments, two gaps remain particularly relevant for ultra-high-dimensional
scientific data. First, canonical global--local priors are exchangeable across predictor
coordinates before observing the likelihood and therefore do not directly exploit the
inexpensive marginal evidence that motivates successful screening procedures. Conversely,
screening and guided penalization methods typically use such information outside the
Bayesian shrinkage mechanism, either to remove variables or to construct a separate
penalized estimation step. This suggests an intermediate strategy: use marginal evidence
continuously to modulate shrinkage, rather than as a binary screening decision, while
retaining joint likelihood-based learning for all regression coefficients. Second, posterior
computation itself becomes a serious obstacle as $p$ grows. Fast Gaussian-scale-mixture
sampling substantially reduces the cost of coefficient updates
\citep{bhattacharya2016fast}, and scalable approximations have extended horseshoe
computation to much larger problems \citep{johndrow2020scalable}; nevertheless,
repeatedly updating predictor-specific local parameters can remain costly when the number of
features reaches hundreds of thousands or more.

In this paper, we propose \emph{Bayesian Univariate-Guided Sparse Regression} (BUGS), a
global--local shrinkage framework that incorporates marginal relevance information directly
into the prior through a continuous modulation of shrinkage. For each predictor, a
stabilized and bounded marginal association score is used as a guidance statistic and
embedded within the nonlinear variance mapping of a regularized horseshoe prior. Predictors
with stronger marginal evidence are thereby allowed to escape global shrinkage more readily,
whereas those with weaker evidence receive stronger shrinkage; importantly, marginal
guidance alone does not remove predictors from the joint regression model. The strength of
the guidance is itself learned from the data, and the construction reduces to the regularized
horseshoe when the guidance effect vanishes. In contrast to simply multiplying an otherwise
fixed prior variance by predictor-specific weights, the guidance enters inside the nonlinear
regularized-horseshoe variance map and therefore modifies the transition between strong
shrinkage and slab-like behavior.

\begin{figure}[!t]
    \centering
    \newlength{\figwidth}
    \setlength{\figwidth}{0.84\linewidth}
    \includegraphics[width=\figwidth]{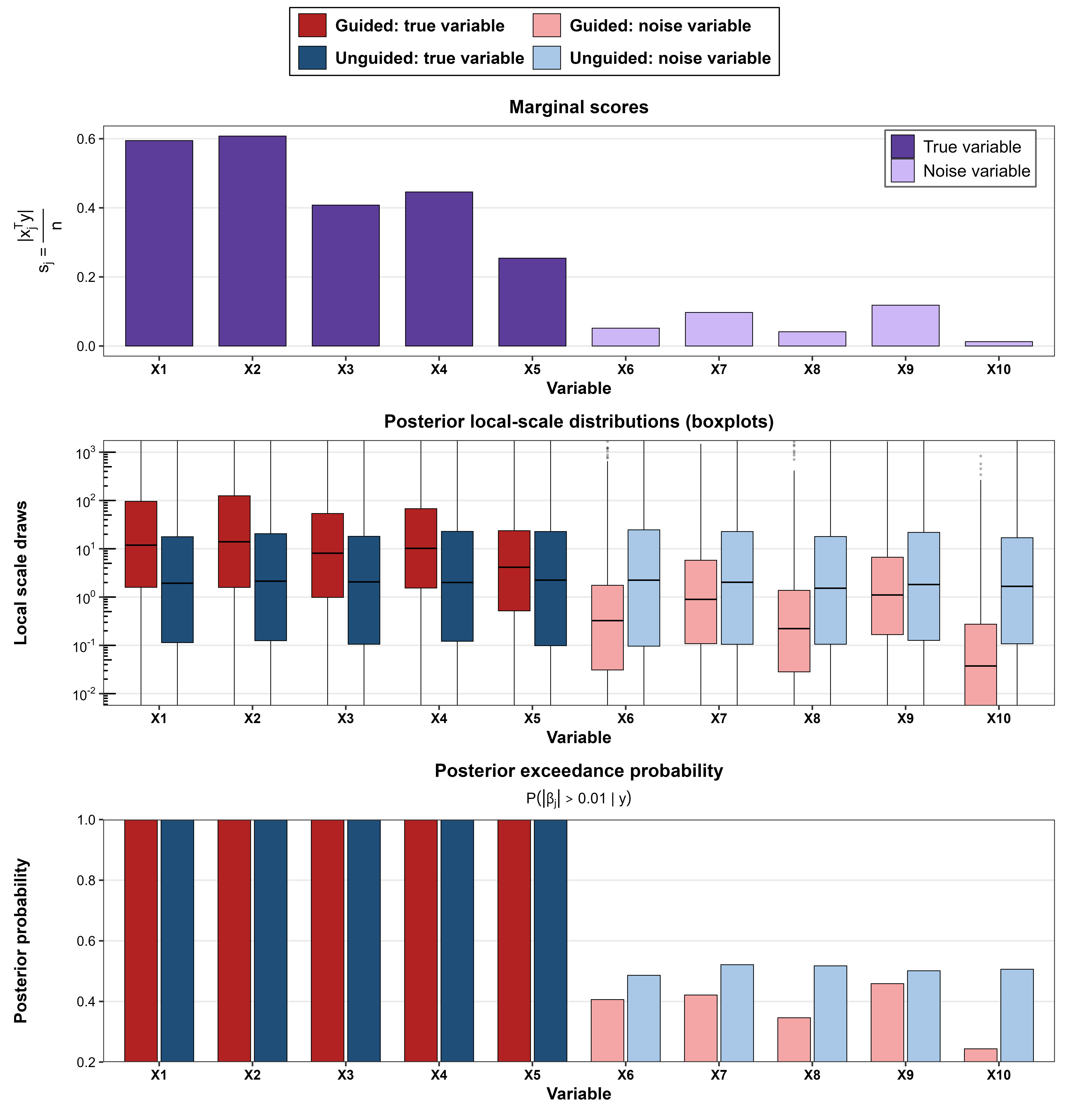}
    \caption{Illustration of marginal guidance in a sparse regression example with $n=100$ and $p=10$. The guided prior increases the effective shrinkage scale for
    signal variables while maintaining stronger shrinkage of noise variables relative to the regularized horseshoe, resulting in sharper separation in posterior
    exceedance probabilities.}
    \label{fig:motivational_plot}
    \vspace{-0.5cm}
\end{figure}

Figure~\ref{fig:motivational_plot} illustrates the resulting mechanism in a simple sparse
regression example (see Supplementary Section S1 for details). The guided prior produces
clearer separation between signal and noise variables than its unguided counterpart, which
motivates its use for sparse discovery when many candidate predictors are present. We
provide theoretical support for this construction through results on conditional prior
concentration and posterior contraction under standard sparsity conditions
\citep{GhosalGhoshvdV2000, vanderPasKleijnvdV2014, vanderPasSzabovdV2017}.
We further characterize the role of guidance, showing that informative marginal evidence
induces systematic shrinkage separation between active and inactive predictors, whereas
uninformative guidance does not itself create such separation.

To make the framework computationally feasible in ultra-high dimensions, we additionally
develop BUGS-Active, an active-set MCMC approximation in which predictor-specific local
scales are updated only for a data-adaptive subset $A_n$, while regression coefficients and
global parameters continue to be updated jointly. This reduces the cost of local-scale
updates from $O(p)$ to $O(|A_n|)$ when $|A_n|\ll p$ and connects computationally with
recent work on scalable Bayesian shrinkage
\citep{bhattacharya2016fast, johndrow2020scalable}. Under suitable screening conditions,
we establish sure-screening and posterior-contraction properties for the corresponding
active posterior. Our posterior contraction results concern fixed or independently
constructed guidance. We assess the practical impact of the default
same-sample construction through a sensitivity comparison with
guidance constructed from a disjoint half-sample.

We evaluate BUGS and BUGS-Active under both independent and correlated predictor
designs, comparing them with a broad collection of frequentist and Bayesian sparse
regression methods across dimensions extending to $p\approx10^6$. The simulations assess
estimation, prediction, signal recovery, false discoveries, and computation, with particular
attention to whether improved specificity can be achieved without sacrificing sensitivity.
Finally, we apply BUGS-Active to longitudinal DNA methylation data from the Growing Up in
Singapore Towards healthy Outcomes (GUSTO) cohort \citep{Leroy2025}. The analysis
contains $n=1051$ samples and $865{,}859$ CpG sites and seeks to identify a sparse set of
methylation sites associated with developmental age while retaining accurate out-of-sample
prediction. All CpG probes are analyzed jointly without preliminary feature screening. This
application illustrates the setting that motivates the proposed framework: an interpretable
Bayesian analysis in which the predictor dimension approaches one million, marginal
information is readily available, and both sparse discovery and prediction are scientifically
relevant.

%%%%%%%%%%%%%%%%%%%%%%%%%%%%%%%%%%%%%%%%%%%%%%%%%%%%%%%%%%%%%%%%%%%%%%%
\section{Marginally Guided Global--Local Shrinkage Prior}
\label{sec:method}
\subsection{Model Setup}
We consider the Gaussian linear model
\begin{equation}
\label{eqn:likelihood}
y = X\beta + \varepsilon, 
\qquad \varepsilon \sim N_n(0,\sigma^2 I_n),
\end{equation}

\noindent where $y \in \mathbb{R}^n$, $X \in \mathbb{R}^{n \times p}$, and $\beta \in \mathbb{R}^p$. We focus on the high-dimensional regime $p \gg n$ under sparsity of $\beta$. Throughout, $y$ and the columns of $X$ are standardized to have mean zero and unit variance, so that regression coefficients are directly comparable across predictors.

\subsection{Univariate Guidance Statistics}
Let $x_j$ denote the $j$th column of $X$. To quantify marginal relevance, we define the score
\begin{equation}
\label{eqn:score}
s_j = \frac{1}{n} \big| x_j^\top y \big|, \qquad j = 1,\dots,p, 
\end{equation}

\noindent which corresponds to the absolute marginal correlation under standardization. We stabilize and standardize these scores via
\[
\tilde z_j = \log(s_j + \epsilon), 
\qquad
z_j = \frac{\tilde z_j - \bar{\tilde z}}{\mathrm{sd}(\tilde z_1,\dots,\tilde z_p)},
\]

\noindent and define the clipped guidance statistic
\[
\tilde z_j^{\,*} = \max\{-C_z, \min(z_j, C_z)\}.
\]

\noindent The quantities $\tilde z_j^{\,*}$ act as covariates encoding marginal evidence, with larger values indicating stronger relevance. Although constructed from the observed data, we treat $\widetilde{\mathbf z}^{\,*}$ as fixed when defining the conditional prior. The posterior contraction results assume fixed or independently
constructed guidance, while the practical impact of constructing
guidance from the same sample is examined empirically in Section~\ref{sec:sim_study}.

Alternative choices of $s_j$ (e.g., marginal $t$-statistics or Bayes factors) can be used without altering the framework, provided they preserve the ordering of marginal evidence \citep{fan2008sis, efron2010large, chatterjee2025unilasso}.

\subsection{Marginally Guided Regularized Horseshoe Prior}
We propose a marginally guided extension of the regularized horseshoe prior in which univariate relevance information is incorporated directly into the global--local shrinkage mechanism; a schematic summary is given in Figure~\ref{fig:prior_flowchart}. Specifically, for each coefficient,
\[
\beta_j \mid \lambda_j, \tau, c, \eta, \sigma^2, \tilde z_j^{\,*}
\sim N\!\left(0,\sigma^2 \tilde\kappa_j^2\right),
\]

\noindent with effective variance
\begin{equation}
\label{eqn:kappa_tilde}
\tilde{\kappa}_j^2
=
\frac{c^2 \tau^2 \lambda_j^2 \exp(\eta \tilde z_j^{\,*})}
     {c^2 + \tau^2 \lambda_j^2 \exp(\eta \tilde z_j^{\,*})}.
\end{equation}

\noindent For notational convenience, define the mapping
\[
\kappa^2(z;\lambda,\tau,c,\eta)
:=
\frac{c^2 \tau^2 \lambda^2 \exp(\eta z)}
     {c^2 + \tau^2 \lambda^2 \exp(\eta z)},
\]

\noindent so that \(\tilde{\kappa}_j^2 = \kappa^2(\tilde z_j^{\,*};\lambda_j,\tau,c,\eta)\). Here, $\lambda_j$ is the local scale, $\tau$ is the global scale, $c$ is the slab regularization parameter, and $\eta \ge 0$ controls the influence of marginal guidance. Larger values of $\tilde z_j^{\,*}$ increase the effective variance and therefore reduce shrinkage, whereas smaller values induce stronger shrinkage. Thus, the proposed construction preserves the regularized horseshoe structure while introducing a continuous guidance mechanism. Figure~\ref{fig:prior_mechanism} illustrates how guidance modifies both the multiplier $\exp(\eta \tilde z_j^{\,*})$ and the resulting effective variance.
\begin{figure}[!ht]
    \centering
    \setlength{\figwidth}{0.85\linewidth}
    \includegraphics[width=\figwidth]{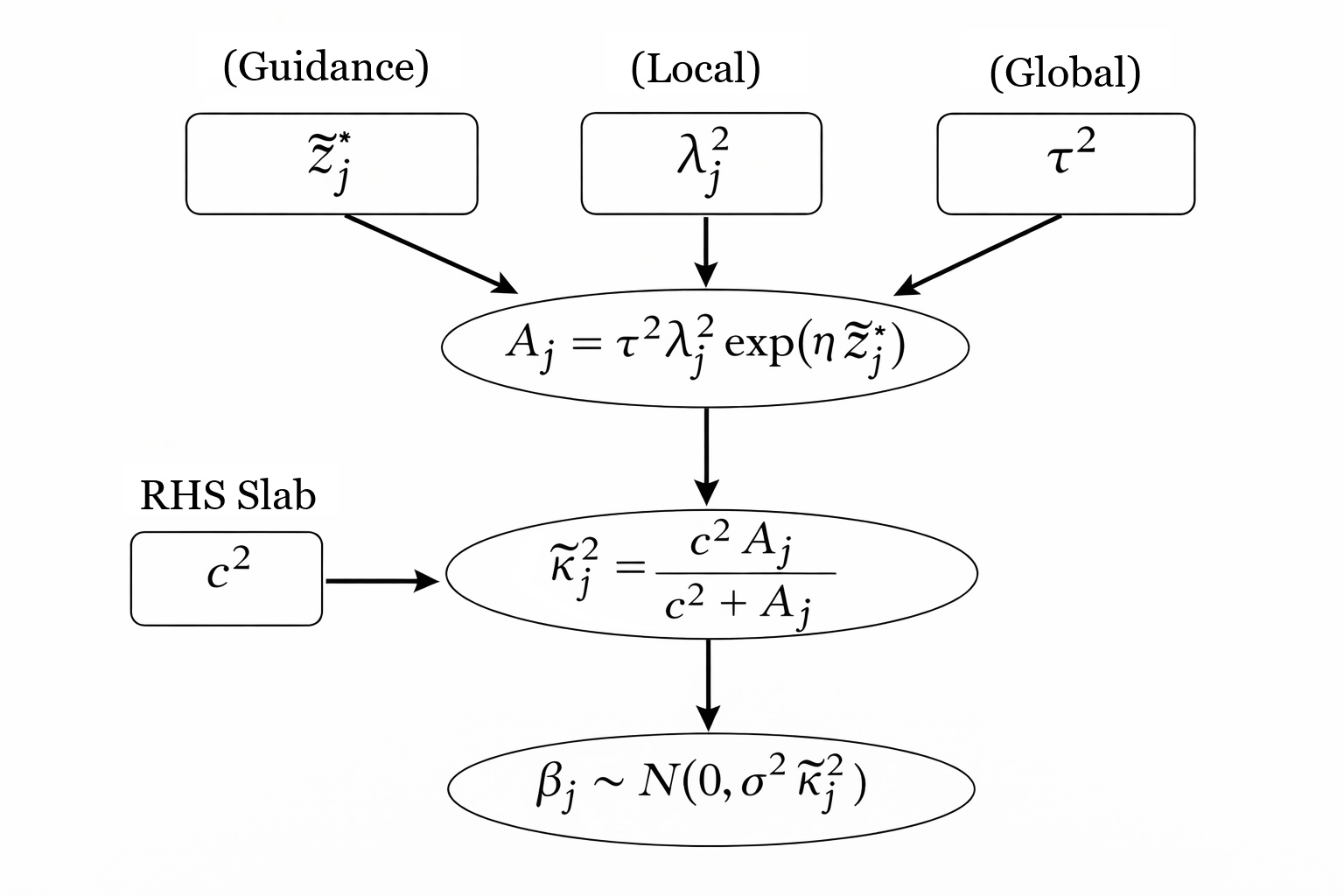}
    \caption{Schematic representation of the guided regularized horseshoe prior. 
        The effective scale $A_j = \tau^2 \lambda_j^2 \exp(\eta \tilde{z}_j^{\,*})$ combines global shrinkage ($\tau^2$), local shrinkage ($\lambda_j^2$), and marginal guidance through the multiplicative factor $\exp(\eta \tilde{z}_j^{\,*})$. 
        The slab parameter $c^2$ then regularizes this scale via $\tilde{\kappa}_j^2 = \frac{c^2 A_j}{c^2 + A_j}$, yielding the effective prior variance of $\beta_j$. Thus, guidance acts multiplicatively on the local--global shrinkage mechanism, while the slab ensures bounded variance for large signals.}
    \label{fig:prior_flowchart}
\end{figure}

\begin{figure}[!t]
    \centering
    \setlength{\figwidth}{0.99\linewidth}
    \includegraphics[width=\figwidth]{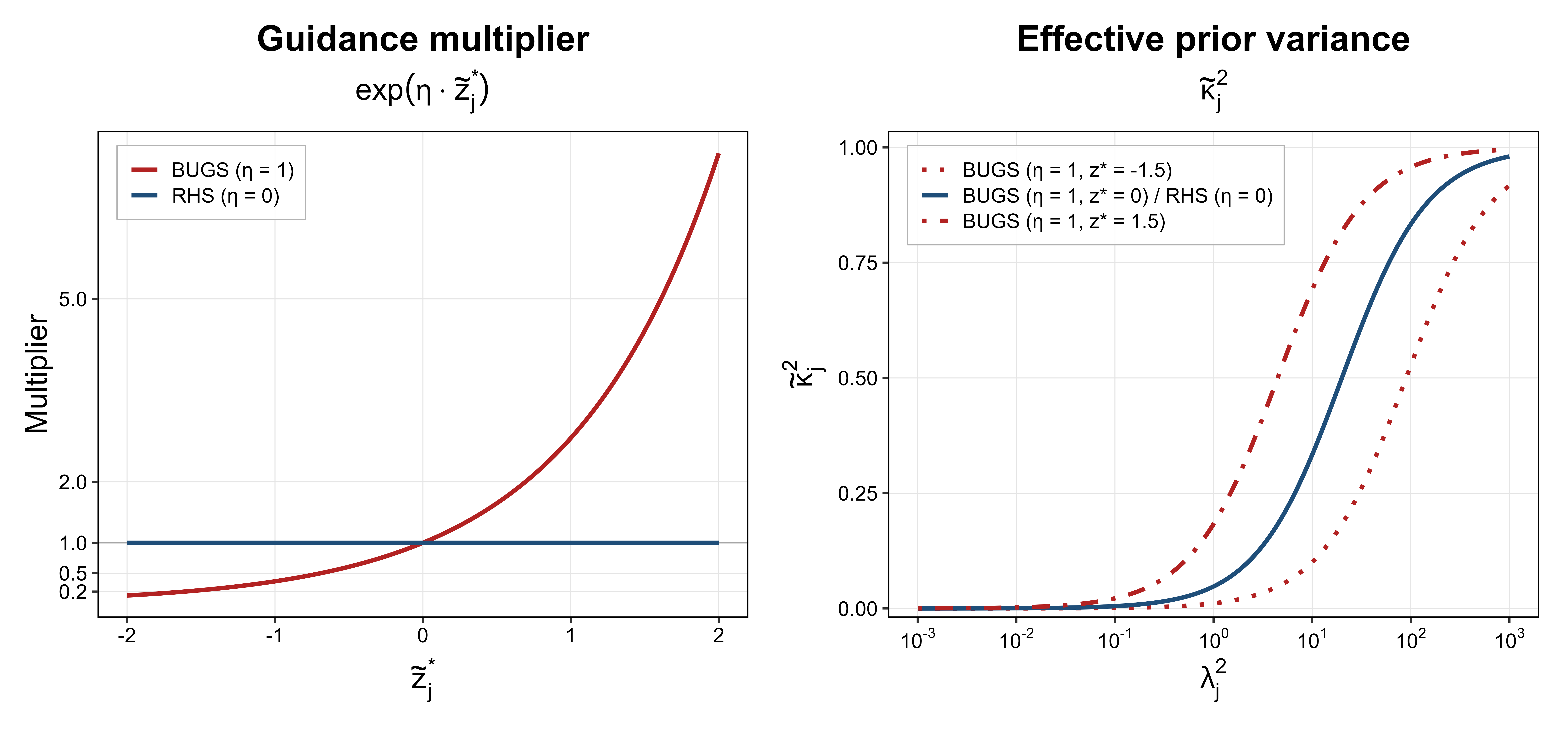}
    \vspace{-0.5cm}
    \caption{Effect of marginal guidance on the prior. \textit{Left:} Guidance multiplier $\exp(\eta \tilde{z}_j^{\,*})$ as a function of the standardized guidance statistic $\tilde{z}_j^{\,*}$, illustrating how positive (negative) values amplify (attenuate) the effective scale relative to the baseline regularized horseshoe ($\eta=0$). \textit{Right:} Induced effective prior variance $\tilde{\kappa}_j^2$ as a function of the local scale $\lambda_j^2$ under different guidance levels. Positive guidance shifts the shrinkage curve left, leading to earlier escape from shrinkage, while negative guidance delays this transition. When $\tilde{z}_j^{\,*}=0$, the prior reduces to the standard regularized horseshoe.}
    \label{fig:prior_mechanism}\vspace{-0.5cm}
\end{figure}

\subsection{Relation to screening and data-guided shrinkage methods}
\label{sec:relation_existing}
The marginal scores $s_j$ in \eqref{eqn:score} are closely related to screening procedures such as sure independence screening \citep{fan2008sis}, which rank predictors using marginal associations. The distinction is that screening methods use such information primarily as a preprocessing device, whereas our approach incorporates it directly into the prior so that guidance acts continuously through shrinkage rather than through hard thresholding.

Because the guidance statistics $\tilde z_j^{\,*}$ are computed from the observed data, the prior is formally data-dependent and is interpreted conditionally on the realized guidance vector. In this sense, it differs from classical empirical-Bayes procedures \citep{robbins1956empirical, efron2010large}, which estimate prior parameters by marginal likelihood optimization. Here, the data enter only through a stabilized and bounded transformation of marginal scores, after which posterior inference proceeds conditionally on the realized guidance values. The proposed prior is closely related to the horseshoe and regularized horseshoe
families \citep{carvalho2010horseshoe, PiironenVehtari2017}, but augments their
predictor-specific shrinkage structure with marginal guidance. The key distinction is
structural: the guidance term is inserted directly within the nonlinear regularized-horseshoe
variance map in~(3), where it interacts jointly with the global scale $\tau$, local scale
$\lambda_j$, and slab parameter $c$. As a consequence, guidance changes not only the magnitude of shrinkage but also the transition between strong shrinkage and slab-like behavior, effectively shifting the shrinkage threshold. The proposed prior may therefore be viewed as a continuous, guidance-informed generalization of marginal screening ideas within a global--local shrinkage framework.
\subsection{Prior hierarchy and interpretation}

The prior hierarchy is given by
\begin{equation}
\label{eqn:hierarchy}
\lambda_j \sim C^+(0,1),\;
\tau \sim C^+(0,\tau_0),\;
c^2 \sim \mathrm{IG}(a_c,b_c),\;
\eta \sim N^+(0,\sigma_\eta^2),\;
\sigma^2 \sim \mathrm{IG}(a_\sigma,b_\sigma),
\end{equation}

\noindent independently for $j=1,\dots,p$, where $C^+$ and $N^+$ denote half-Cauchy and half-normal distributions, respectively.

The model defines a univariate-guided global--local shrinkage prior. The global parameter $\tau$ controls overall sparsity, while the local scales $\lambda_j$ enable predictor-specific adaptation; the slab parameter $c$ stabilizes large coefficients. Marginal guidance enters multiplicatively through $\exp(\eta \tilde{z}_j^{\,*})$, modulating the effective shrinkage scale, with $\eta$ learned from the data. When $\eta=0$, the prior reduces to the regularized horseshoe. The prior remains centered at zero for all $\tilde{z}_j^{\,*}$, so that guidance affects only the degree of shrinkage, not the direction of effects. The joint posterior (conditional on $\tilde{z}^{\,*}$) is
\begin{align}
\label{eqn:posterior}
\pi(\beta,\lambda,\tau,c^2,\eta,\sigma^2 \mid y,X,\tilde{z}^{\,*})
&\propto
f(y \mid X,\beta,\sigma^2)
\prod_{j=1}^p 
\pi(\beta_j \mid \lambda_j,\tau,c^2,\eta,\sigma^2,\tilde{z}_j^{\,*})
\pi(\lambda_j) \nonumber \\
&\quad \times \pi(\tau)\pi(c^2)\pi(\eta)\pi(\sigma^2).
\end{align}

%%%%%%%%%%%%%%%%%%%%%%%%%%%%%%%%%%%%%%%%%%%%%%%%%%%%%%%%%%%%
\section{Posterior Computation}
\label{sec:computation}
Posterior inference is performed via a Markov chain Monte Carlo (MCMC) algorithm combining Gibbs updates for conjugate components with slice sampling for non-conjugate parameters \citep{neal2003slice}. The Gaussian likelihood together with the shrinkage prior yields closed-form updates for the regression coefficients and noise variance, while the remaining parameters are updated using univariate slice sampling on log-transformed scales.
\subsection{Sampling Scheme}
Let $D = \mathrm{diag}(\tilde{\kappa}_1^2,\dots,\tilde{\kappa}_p^2)$, where $\tilde{\kappa}_j^2$ is defined in \eqref{eqn:kappa_tilde}. Each MCMC iteration proceeds by sequentially updating $(\beta, \sigma^2, \lambda, \tau, c^2, \eta)$ from their full conditionals.
\noindent \textbf{Update of $\beta$.} Conditional on the remaining parameters, the posterior of $\beta$ is Gaussian,
\[
\beta \mid \cdot \sim N_p\!\left(m_\beta, V_\beta\right),
\quad
V_\beta = \sigma^2 (X^\top X + D^{-1})^{-1}, 
\quad
m_\beta = (X^\top X + D^{-1})^{-1} X^\top y.
\]

\noindent Direct sampling is infeasible when $p \gg n$ due to inversion of a $p \times p$ matrix. We therefore adopt the fast sampling scheme of \cite{bhattacharya2016fast}, based on the Woodbury identity,

\[
(X^\top X + D^{-1})^{-1}
= D - D X^\top (I_n + X D X^\top)^{-1} X D,
\]

\noindent which reduces the computational cost from $O(p^3)$ to $O(n^2 p + n^3)$. This enables efficient sampling by solving an $n \times n$ linear system, making the update scalable in high dimensions.

\noindent \textbf{Update of $\sigma^2$.}
The noise variance admits a conjugate inverse-gamma update,

\[
\sigma^2 \mid \cdot \sim \mathrm{IG}\!\left(
a_\sigma + \frac{n+p}{2},\;
b_\sigma + \frac{1}{2}\|y - X\beta\|_2^2
+ \frac{1}{2}\sum_{j=1}^p \frac{\beta_j^2}{\tilde{\kappa}_j^2}
\right).
\]

\noindent \textbf{Updates of $(\lambda_j), \tau, c^2, \eta$.}
The remaining parameters enter the model through the nonlinear shrinkage factor $\tilde{\kappa}_j^2$ and are updated via univariate slice sampling \citep{neal2003slice} on suitable log-transformed scales. In particular, their conditional densities arise from the dependence of $\tilde{\kappa}_j^2$ on $(\lambda_j,\tau,c^2,\eta)$. We update $\lambda_j$, $\tau$, and $c^2$ on $\log$-scales, and $\eta$ on $[0,\infty)$, using their respective conditionals. Details are provided in Supplementary Section S2.

%%%%%%%%%%%%
\subsection{Active-set MCMC approximation for ultra-high dimensions}
\label{subsec:BUGS_active_sampling}

For moderate to high-dimensional settings (e.g., \(p \lesssim 10^4\)), the full MCMC sampler targets the exact posterior. In ultra-high-dimensional regimes (\(p \gg 10^4\)), however, updating all local shrinkage parameters \(\{\lambda_j\}_{j=1}^p\) at each iteration becomes computationally prohibitive. To address this, we introduce an active-set MCMC approximation that restricts local updates to a subset of coordinates likely to be influential. This exploits a key property of global--local shrinkage priors: most coefficients are strongly shrunk toward zero, so updating their associated local scales has negligible impact on posterior inference. Related ideas have been explored in scalable Bayesian regression \citep{johndrow2020scalable}.

At each MCMC iteration, we construct an active set \(A_n \subset \{1,\dots,p\}\), which may vary across iterations, consisting of coordinates deemed potentially relevant based on the current state of the chain. A coordinate \(j\) is included in \(A_n\) if it satisfies either of the following:
\begin{enumerate}
\item it belongs to a fixed-size subset of predictors with the largest marginal guidance values \(\tilde{z}_j^{\,*}\) (the \emph{guidance budget}), or
\item its current coefficient magnitude exceeds a threshold, \(|\beta_j| > t_n\).
\end{enumerate}

\noindent The active set may additionally be capped to include only the largest coefficients by magnitude to control computational cost. Given \(A_n\), local shrinkage parameters are updated only for indices in \(A_n\), while coefficients are updated globally.

The MCMC updates then proceed as follows. The regression coefficients \(\beta\) and global parameters \((\tau, c^2, \eta, \sigma^2)\) are updated as in the full sampler. The local shrinkage parameters \(\{\lambda_j\}\) are updated only for \(j \in A_n\), while for \(j \notin A_n\) they are fixed at a small baseline value, enforcing strong shrinkage. This yields a computationally efficient approximation that reduces the cost of local updates from \(O(p)\) to \(O(|A_n|)\) per iteration, while retaining global updates for \(\beta\) so that all coordinates continue to interact through the likelihood.

We emphasize that this is an approximate MCMC scheme tailored for ultra-high-dimensional settings rather than an exact sampler. The active-set construction is designed to approximate the screening mechanism analyzed in Section~\ref{subsec:BUGS_active}. In particular, the guidance-based and coefficient-based inclusion rules mirror the signal detection and separation conditions in Assumptions (B1)--(B2). The corresponding theoretical analysis studies an idealized active posterior motivated by these inclusion rules and therefore provides support for the screening principle underlying BUGS-Active, rather than an exact finite-sample characterization of the adaptive MCMC approximation.

%%%%%%%%%%%%%%%%%%%%%%%%%%%%%%%%%%%%%%%%%%%%%%%%%%%%%%%%%%%%%%%%%%%%%%%%%%%%%%%%%%
\section{Theoretical Properties}\label{sec:theory}
This section establishes theoretical guarantees for the proposed guided shrinkage framework under the stated assumptions. We first analyze the full model under the guided regularized horseshoe prior, deriving prior concentration and posterior contraction results. We then study an idealized active posterior, establishing sure-screening and posterior contraction results. These guarantees apply within the regime specified by the design and regularity assumptions; the empirical study additionally explores ultra-high-dimensional settings beyond the scope of these guarantees.

Detailed proofs and full assumption statements are provided in Supplementary Section S3, while brief proof sketches are included for the main results. 
Throughout, prior statements are formulated conditionally on the clipped
guidance vector
\[
\tilde z^{\,*}=(\tilde z_1^{\,*},\dots,\tilde z_p^{\,*})^\top.
\]
For the posterior contraction results, this vector is either deterministic
or constructed from auxiliary data independent of the likelihood sample,
conditional on the design matrix. In the auxiliary-data formulation, the
testing and prior-mass conditions are imposed conditionally on the
auxiliary information, with constants and probability bounds uniform
over its admissible realizations. For the active-posterior contraction
result, the screening statistics and active set are also constructed
from this independent auxiliary information. These results do not
directly establish contraction for the default same-sample guidance
procedure.

\subsection{Properties of the marginally guided shrinkage prior (BUGS)}\label{subsec:BUGS}
We establish prior concentration and posterior contraction
properties of the proposed model in high-dimensional regimes.
The analysis is conducted under standard sparsity, design,
and regularity conditions, together with assumptions
characterizing the role of marginal guidance.
Let $S_0=\{j\in\{1,\dots,p\}:\beta_{0j}\neq0\}$ denote
the true support, with $s_0=|S_0|$. We assume:

\begin{itemize}
\item[(A1)] Sparsity: $s_0 \log p = o(n)$ and $s_0 \log p \to \infty$.
\item[(A2)] Design: $X$ satisfies a restricted eigenvalue condition \citep{vanderPasKleijnvdV2014}.
\item[(A3)] Gaussian noise: $y = X\beta_0 + \varepsilon$, $\varepsilon \sim N_n(0,\sigma_0^2 I_n)$.
\item[(A4)] Bounded guidance: $\max_j |\tilde z_j^{\,*}| \le C_z$.
\item[(A5)] Bounded signals: $\max_{j\in S_0} |\beta_{0j}| \le B$.
\item[(A6)] Hyperpriors assign positive mass to neighborhoods of admissible values.
\end{itemize}
These conditions are standard in high-dimensional Bayesian regression and ensure well-behaved likelihood geometry and sufficient prior mass in relevant regions.

To characterize the effect of marginal guidance, we next consider two complementary regimes, corresponding to the absence and presence of useful marginal signals:
\begin{itemize}
\item[(A7)] Uninformative guidance: the guidance statistics $\tilde z_j^{\,*}$ do not asymptotically distinguish active and inactive coordinates.
\item[(A8)] Informative guidance: the guidance statistics $\tilde z_j^{\,*}$ exhibit a fixed separation between active and inactive coordinates.
\end{itemize}
These conditions are imposed separately in the corresponding theoretical results and are not assumed to hold simultaneously. We further impose the following technical conditions, which are standard in posterior contraction analyses:
\begin{itemize}
\item[(A9)] Dimension growth condition relating $n$ and $p$.
\item[(A10)] Standard testing and sieve conditions for posterior contraction 
\citep{GhosalGhoshvdV2000, vanderPasKleijnvdV2014, vanderPasSzabovdV2017}.
\item[(A11)] Bounded design operator norm: $\|X\|_{\mathrm{op}} = O(\sqrt{n})$.
\end{itemize}

\noindent A complete formal statement of these assumptions, together with brief justifications and the corresponding proof details, is provided in Supplementary Section S3.

We first establish a prior concentration result, showing that the proposed prior assigns sufficient mass to an \(\ell_2\)-neighborhood of the true sparse coefficient vector.

\begin{theorem}[Conditional prior support]
\label{thm:prior_mass}
Assume the prior hierarchy in \eqref{eqn:hierarchy}, and suppose that
(A1), (A4), (A5), (A6), and (A9) hold. Let

\[
\epsilon_n = C \sqrt{\frac{s_0 \log p}{n}}
\]

\noindent for a sufficiently large constant \(C>0\). Then there exists a constant
\(C_1>0\), independent of \(n\) and \(p\), such that

\[
\Pi\!\left(
\|\beta-\beta_0\|_2 \le \epsilon_n
\,\middle|\,
\tilde z^{\,*}
\right)
\ge
\exp(-C_1 s_0 \log p)
\]

\noindent for all sufficiently large \(n\).
\end{theorem}

We now establish posterior contraction under the fixed or independently
constructed guidance formulation described above. We work under the
Gaussian linear model with known error variance $\sigma_0^2$ and use the
coefficient prior induced by the hierarchy with $\sigma^2$ fixed at
$\sigma_0^2$, integrating over the remaining hyperparameters. The result
is conditional on the testing and sieve requirements in (A10).

\begin{theorem}[Posterior contraction]
\label{thm:contraction}
Consider the Gaussian linear model with known variance $\sigma_0^2$.
Let the guidance vector be deterministic or constructed from auxiliary
data independent of the likelihood sample, as specified above.
Suppose that Assumptions (A1)--(A6), (A9), (A10), and (A11) hold, with
(A10) applying to the induced coefficient prior at
$\sigma^2=\sigma_0^2$.
Then, for a sufficiently large constant $M>0$,

\[
\Pi\!\left(
\|\beta - \beta_0\|_2 \ge M \epsilon_n
\;\middle|\; y, \tilde z^{\,*}
\right)
\to 0
\quad \text{in probability under } P_{\beta_0},
\]

\noindent where $\epsilon_n^2 = \frac{s_0 \log p}{n}$.
\end{theorem}

We next examine the behavior of the proposed prior under uninformative guidance, where the marginal guidance statistics do not asymptotically distinguish active and inactive coordinates. The following result shows that, in this regime, the guidance component alone does not induce shrinkage separation.

\begin{proposition}[No separation under uninformative guidance]
\label{prop:no_sep}
Suppose that Assumption (A7) holds. Fix \(\tau>0\), \(c>0\),
\(\lambda>0\), and \(\eta \ge 0\). Then

\[
\sup_{j\in S_0,\;k\notin S_0}
\left|
\kappa^2(\tilde z_j^{\,*};\lambda,\tau,c,\eta)
-
\kappa^2(\tilde z_k^{\,*};\lambda,\tau,c,\eta)
\right|
\to 0.
\]

\noindent Consequently, when the guidance statistics are asymptotically uninformative,
the guidance component alone does not induce shrinkage separation between
active and inactive variables.
\end{proposition}

We next consider the complementary regime in which the guidance statistics asymptotically separate active and inactive coordinates. In contrast to the uninformative case, the guidance component induces systematic shrinkage separation, favoring reduced shrinkage for active variables and increased shrinkage for inactive ones.

\begin{proposition}[Guidance-induced shrinkage separation]
\label{prop:separation}
Suppose that Assumption (A8) holds. Fix \(\tau>0\), \(c>0\),
\(\lambda>0\), and \(\eta > 0\). Define

\[
\kappa^2(z;\lambda,\tau,c,\eta)
:=
\frac{c^2\tau^2\lambda^2 e^{\eta z}}
{c^2+\tau^2\lambda^2 e^{\eta z}}.
\]
\noindent Then
\[
\kappa^2(\tilde z_j^{\,*};\lambda,\tau,c,\eta)
>
\kappa^2(\tilde z_k^{\,*};\lambda,\tau,c,\eta)
\quad
\text{for all } j \in S_0,\; k \notin S_0.
\]

\noindent Consequently, under informative guidance, the guidance component induces
systematically weaker shrinkage for active variables relative to inactive
variables.
\end{proposition}

\noindent This separation property motivates the active-set construction analyzed in the next section, where guidance is used to identify a reduced set of candidate variables.
%%%%%%%%%%%%%%%%%%%%%%%%%%%%%%%%%%%%%%%%%%%%%%%%%%%%%%%%%%

\subsection{Active-set posterior: screening and contraction}\label{subsec:BUGS_active}
We study the statistical properties of the active-set approximation
introduced in Section~\ref{subsec:BUGS_active_sampling}.
Let $A_n\subseteq\{1,\dots,p\}$ denote an active set, and define
\[
\Theta(A_n)
:=
\{\beta\in\mathbb R^p:
\beta_j=0 \text{ for } j\notin A_n\}.
\]
Let $\Pi_{A_n}(\cdot\mid\tilde z^{\,*})$ denote the prior
obtained by applying the hierarchy to the coordinates in
$A_n$, retaining the original hyperpriors and fixing
$\beta_{A_n^c}=0$. The active posterior
$\Pi_{A_n}(\cdot\mid y,\tilde z^{\,*})$ is obtained from
this reduced prior and the likelihood restricted to
$\Theta(A_n)$.

For the posterior contraction result below, the active set,
screening statistics, and guidance vector are constructed
from auxiliary data independent of the likelihood sample,
conditional on the design matrix.

\noindent Our goal is to show that this restriction retains the true support with high probability and preserves posterior contraction around $\beta_0$.

We impose the following conditions on the active-set construction. Let $\hat\beta^{\,\mathrm{init}}$ denote a generic screening statistic computed from the independent auxiliary sample, and let $t_n,u_n$ denote threshold sequences.
\begin{itemize}
\item[(B1)] \textit{Screening rule.}
The active set satisfies
\[
A_n \supseteq \{j : |\hat\beta^{\,\mathrm{init}}_j| > t_n\}
\cup
\{j : \tilde z_j^{\,*} > u_n\}.
\]
\item[(B2)] \textit{Signal detection.}
There exist sequences $r_n,q_n$ with \(t_n \ll r_n\) and \(u_n \ll q_n\) such that
\[
P\!\left(\min_{j \in S_0} |\hat\beta^{\,\mathrm{init}}_j| \ge r_n\right) \to 1,
\qquad
P\!\left(\min_{j \in S_0} \tilde z_j^{\,*} \ge q_n\right) \to 1,
\]

\noindent ensuring that true signals exceed the screening thresholds with high probability.

\item[(B3)] \textit{Active-set size control.}
There exists a sequence $L_n$, growing at most polylogarithmically in $p$, such that

\[
P\!\left(|A_n| \le C s_0 L_n\right)\to1.
\]

\item[(B4)] \textit{Uniform testing over admissible active sets.}
Let $m_n = C s_0 L_n$ and $\bar\epsilon_n^2 = (s_0\log m_n)/n$. For sufficiently large $M$, there exists a test $\phi_n$ such that

\[
P_{\beta_0}(\phi_n)\to 0,
\quad
\sup_{\substack{A\supseteq S_0,\ |A|\le m_n}}
\sup_{\substack{\beta\in\Theta(A):\\ \|\beta-\beta_0\|_2 \ge M\bar\epsilon_n}}
P_\beta(1-\phi_n)
\le
\exp(-c_T(M)n\bar\epsilon_n^2).
\]
\end{itemize}

Assumptions (B1)--(B3) formalize a screening regime in which the active set captures all relevant variables while controlling its size. Assumption (B4) ensures that uniform testing arguments extend over admissible active subspaces. These conditions are standard in analyses of screening procedures and sure screening properties \citep{fan2008sis}, and are used here to characterize an idealized active-set mechanism. A complete formal specification of these assumptions and the corresponding proof details is provided in Supplementary Section S3.

In the practical BUGS-Active algorithm, the active set is constructed using quantities derived from current posterior samples, such as coefficient magnitudes and guidance scores. While these do not correspond to any fixed screening estimator, they are expected to exhibit analogous separation behavior once the posterior concentrates around the true sparse signal. Accordingly, Assumptions (B1)--(B2) describe an idealized screening regime that the algorithm approximates, rather than properties of a specific estimator. Assumption (B3) controls the size of the active set, ensuring that the resulting dimension reduction is nontrivial and that the active posterior contracts at a reduced-dimensional rate. Assumption (B4) extends standard testing conditions to the active-set setting, requiring uniform distinguishability of alternatives over admissible subspaces of controlled size. This reflects that the active posterior is supported on a random reduced subspace rather than the full parameter space.

The following result formalizes the corresponding sure-screening property; see \cite{fan2008sis} for related results in high-dimensional screening.

\begin{proposition}[Sure screening of the active set]
\label{prop:screening}
Suppose that Assumptions (B1) and (B2) hold. Then

\[
P_{\beta_0}\!\left(S_0 \subseteq A_n\right) \to 1.
\]
\end{proposition}

%%%%%%%
The following lemma establishes a uniform sieve condition over admissible active subspaces, which, together with the screening property, enables control of posterior mass on reduced-dimensional parameter spaces.

\begin{lemma}[Uniform sieve condition on admissible active subspaces]
\label{lem:active_sieve}
Consider the Gaussian linear model with known variance \(\sigma_0^2\). Suppose that Assumptions (A1)--(A6), (A9), and (A11) hold. Let

\[
m_n := C s_0 L_n,
\qquad
\bar\epsilon_n^2 := \frac{s_0\log m_n}{n},
\]

\noindent where \(C\) and \(L_n\) are as in Assumption (B3), and define

\[
\mathcal A_n
:=
\Bigl\{
A \subseteq \{1,\dots,p\} :
S_0 \subseteq A,\;
|A|\le m_n
\Bigr\}.
\]

\noindent For each \(A\in\mathcal A_n\), let

\[
\Theta(A)
:=
\{\beta\in\mathbb R^p:\beta_j=0\ \text{for all } j\notin A\},
\]

\noindent and let \(\Pi_A(\cdot\mid \tilde z^{\,*})\) denote the induced guided prior on \(\Theta(A)\), obtained by restricting the hierarchical prior to the coordinates in \(A\) and fixing \(\beta_{A^c}=0\). Then, for every fixed constant $C_F>0$ and each $A\in\mathcal A_n$, there exists a measurable set $\mathcal F_{n,A}\subseteq\Theta(A)$
satisfying

\[
\Pi_A(\mathcal F_{n,A}^c\mid \tilde z^{\,*})
\le
\exp(-C_F n\bar\epsilon_n^2)
\]

\noindent for all sufficiently large \(n\).
\end{lemma}
%%%%%%%
For the denominator bound and active-posterior contraction result,
we additionally impose a uniform reduced-space prior concentration
condition. With $m_n$, $\bar\epsilon_n$, and $\mathcal A_n$ defined
above, assume that there exists a constant $C_1>0$, independent
of $n$, the admissible active set $A$, and the admissible guidance
realization, such that
\begin{equation}
\label{eq:uniform-reduced-prior-mass}
\Pi_A\!\left(
\|\beta-\beta_0\|_2 \le \bar\epsilon_n
\,\middle|\,
\tilde z^{\,*}
\right)
\ge
\exp(-C_1 n\bar\epsilon_n^2)
\end{equation}
for all sufficiently large $n$, uniformly over
$A\in\mathcal A_n$ and admissible guidance realizations.
Here $\Pi_A$ denotes the reduced prior with
$\sigma^2=\sigma_0^2$.

Under this additional condition, the following lemma provides
a uniform lower bound on the restricted marginal likelihood.
\begin{lemma}[Uniform denominator lower bound]
\label{lem:active_denom}
Assume the Gaussian linear model with known variance $\sigma_0^2$.
Suppose that (A1) and (A11) hold, that $s_0\log m_n\to\infty$, and
that the uniform reduced-space prior concentration condition
\eqref{eq:uniform-reduced-prior-mass} holds for the induced prior
at $\sigma^2=\sigma_0^2$. Then there exists a constant $C_D>0$
such that, for every $\delta>0$, there are events $\Omega_{n,A}$
satisfying

\[
\sup_{A\in\mathcal A_n} P_{\beta_0}(\Omega_{n,A}^c)\to 0,
\]

\noindent and on \(\Omega_{n,A}\),

\[
D_{n,A}
\ge
\exp \bigl(-(C_D+\delta)n\bar\epsilon_n^2\bigr),
\]

\noindent where

\[
D_{n,A}
:=
\int_{\Theta(A)}
\frac{f_\beta(y)}{f_{\beta_0}(y)}
\,d\Pi_A(\beta\mid \tilde z^{\,*}).
\]
\end{lemma}
Lemma~\ref{lem:active_denom} provides a denominator lower bound uniformly over admissible active sets. In the contraction result below, this bound is applied to the realized active set \(A_n\). On the high-probability event that \(A_n \in \mathcal A_n\), it yields the corresponding lower bound for the active posterior.

For the following theoretical result, we consider the active posterior conditional on a
screening $\sigma$-field under which $A_n$ and the realized guidance vector
$\widetilde{\mathbf z}^{\,*}$ are fixed. We assume that this screening information is
constructed from an auxiliary sample independent of the likelihood sample used for
posterior updating.

\begin{theorem}[Posterior contraction under the active posterior]
\label{thm:active_contraction}
Consider the Gaussian linear model with known variance $\sigma_0^2$.
Suppose that the active set and guidance are constructed from auxiliary
data independent of the likelihood sample, and that Assumptions
(A1)--(A6), (A9), (A11), and (B1)--(B4) hold under the conditional
formulation specified above. Let
\[
m_n=Cs_0L_n,
\qquad
\bar\epsilon_n^2=\frac{s_0\log m_n}{n},
\]
where $C$ and $L_n$ are as in (B3). Assume additionally that
$s_0\log m_n\to\infty$ and that the uniform reduced-space prior
concentration condition
\eqref{eq:uniform-reduced-prior-mass} holds.
Let $C_D$ be the denominator exponent in
Lemma~\ref{lem:active_denom}, and suppose that the testing and
sieve bounds hold with
\[
\min\{c_T(M),C_F\}>C_D
\]
for all sufficiently large $M$.
Then, for all sufficiently large $M$,
\[
\Pi_{A_n}\!\left(
\|\beta-\beta_0\|_2\ge M\bar\epsilon_n
\,\middle|\,y,\tilde z^{\,*}
\right)
\longrightarrow0
\qquad\text{in }P_{\beta_0}\text{-probability}.
\]
\end{theorem}

Under the stated assumptions, Theorem~\ref{thm:active_contraction} shows that the idealized active posterior achieves a contraction rate of order
\[
\sqrt{\frac{s_0\log m_n}{n}},
\]
retaining the same sparsity-dependent form as the full-posterior rate, with the ambient dimension $p$ replaced by the active-set size bound $m_n$. This result supports the screening principle underlying BUGS-Active within the scope of the stated assumptions.
%%%%%%%%%%%%%%%%%%%%%%%%%%%%%%%%%%%%%%%%%%%%%%%%%%%%%%%%%%%%%%%%%%%

\section{Simulation study}\label{sec:sim_study}
We evaluate the proposed method in terms of estimation, prediction, and variable selection under independent and correlated designs.

\noindent \textbf{Scenario 1 (independent design).}
The predictors $X \in \mathbb{R}^{n \times p}$ are generated with independent standard normal entries and standardized column-wise. The regression vector $\beta$ is sparse with $s_0 = 10$ nonzero coefficients in the first ten coordinates, with values

\[
(2.5,\,-2.0,\,1.8,\,-1.5,\,1.2,\,1.0,\,-0.9,\,0.8,\,0.7,\,-0.7).
\]

\noindent The response is generated as $y = X\beta + \varepsilon$, where $\varepsilon \sim N_n(0, \sigma^2 I_n)$ with $\sigma = 1$.

\noindent \textbf{Scenario 2 (correlated design).}
Predictors are generated with a Toeplitz correlation structure with parameter $\rho = 0.5$ and standardized as above. The regression coefficients and noise model are identical to Scenario 1.

We consider high-dimensional settings with $(p,n)$ ranging up to $p \approx 10^6$. Results are averaged over $10$ replications, with standard errors reported in parentheses.

We compare the proposed BUGS method with a range of widely used frequentist and Bayesian approaches for sparse high-dimensional regression. As frequentist baselines, we include the LASSO implemented via the \texttt{glmnet} package \citep{friedman2010glmnet} and the univariate-guided LASSO (UniLASSO) of \cite{chatterjee2025unilasso}, implemented via the \texttt{uniLasso} R package, which incorporates marginal guidance information into penalized estimation. To directly compare against screening-based strategies in the ultra-high-dimensional regime, we additionally consider sure independence screening (SIS; \citealp[]{fan2008sis}) followed by LASSO, SCAD, or a regularized horseshoe prior, denoted SIS-LASSO, SIS-SCAD, and SIS-RHS, respectively. Following \cite{fan2008sis}, SIS ranks predictors according to their marginal association with the response and retains $d=\lfloor n/\log n\rfloor$ predictors before fitting the corresponding sparse regression model. Among Bayesian methods, we consider several global--local shrinkage priors and related models. The Bayesian LASSO of \cite{park2008blasso}, along with the horseshoe \citep{carvalho2010horseshoe} and horseshoe+ \citep{bhadra2017horseshoeplus} priors, are implemented using the MATLAB \texttt{bayesreg} toolbox \citep{bayesreg_matlab}. We also include the Dirichlet--Laplace prior \citep{bhattacharya2015dl} and the R2D2 prior \citep{zhang2022r2d2}, both with publicly available R implementations\footnote{\url{https://github.com/yandorazhang/R2D2}}, as well as the spike-and-slab LASSO (SSLASSO) \citep{rockova2018ssl}, implemented via the \texttt{SSLASSO} R package. To account for recent developments in scalable Bayesian inference, we include an approximate MCMC implementation of a horseshoe-type prior following \cite{johndrow2020scalable}, implemented via the \texttt{Mhorseshoe} R package \citep{kang2025mhorseshoe}. Finally, we evaluate the proposed BUGS model alongside a computationally scalable variant, BUGS-Active, which employs an adaptive active-set strategy tailored to ultra-high-dimensional settings. Across all methods, tuning parameters and hyperparameters are set to the default values provided by their respective implementations. For Bayesian procedures (including BUGS and BUGS-Active), posterior inference is based on $5000$ MCMC iterations with $1000$ burn-in. Variable selection is based on posterior summaries using a common thresholding rule applied uniformly across Bayesian methods, ensuring comparability of selection metrics.

In our experiments, all methods are evaluated in moderate and high-dimensional settings, namely $(p,n) = (200,100)$, $(500,150)$, and $(1000,200)$. As dimensionality increases, computationally intensive methods (e.g., Dirichlet--Laplace and R2D2) become infeasible and are excluded. The proposed BUGS method is evaluated up to $p = 10^4$, beyond which the BUGS-Active variant is used. In ultra-high-dimensional regimes ($p = 10^4$ to $10^6$), we additionally evaluate the screening-based SIS-LASSO, SIS-SCAD, and SIS-RHS baselines, along with BUGS-Active. LASSO and UniLASSO are evaluated up to $p=10^5$ but are not pursued at $p=10^6$ because of their increased computational burden at this scale. At $p=10^6$, the comparison therefore includes the three SIS-based procedures and BUGS-Active. Performance is evaluated in terms of estimation, prediction, variable selection, and computational efficiency. Estimation accuracy is measured by RMSE of $\beta$, and prediction accuracy by MSE of the response. Variable selection is assessed using TPR, FPR, FDR, and MCC. Computational efficiency is measured by runtime (seconds).

\begin{table}[!t]
\renewcommand{\arraystretch}{1.15}
\resizebox{0.99\columnwidth}{!}{%
\begin{tabular}{c|l|ccccccc}
\hline
$(p,n)$ & Method & RMSE ($\beta$) & MSE ($y$) & TPR & FPR & MCC & FDR & \begin{tabular}[c]{@{}c@{}}CompTime\\ (sec)\end{tabular} \\ \hline
\multirow{10}{*}{\begin{tabular}[c]{@{}c@{}}$p =200$ \\ $n =100$\end{tabular}} & LASSO & 0.013 (0.0009) & 0.020 (0.0024) & 1.000 (0.0000) & 0.081 (0.0080) & 0.609 (0.0198) & 0.594 (0.0228) & 0.07 (0.003) \\
 & UniLASSO & 0.015 (0.0011) & 0.045 (0.0052) & 0.960 (0.0221) & 0.041 (0.0053) & 0.720 (0.0285) & 0.428 (0.0399) & 0.05 (0.003) \\
 & BayesLASSO & 0.023 (0.0012) & 0.039 (0.0025) & 1.000 (0.0000) & 0.512 (0.0094) & 0.213 (0.0038) & 0.907 (0.0015) & 1.70 (0.018) \\
 & Dirich-Laplace & 0.013 (0.0021) & 0.031 (0.0086) & 1.000 (0.0000) & 0.002 (0.0009) & \textbf{0.980 (0.0080)} & \textbf{0.036 (0.0148)} & 19.15 (1.687) \\
 & R2D2 & \textbf{0.007 (0.0004)} & 0.012 (0.0013) & 1.000 (0.0000) & 0.044 (0.0057) & 0.732 (0.0234) & 0.436 (0.0324) & 25.26 (1.082) \\
 & SSLASSO & 0.014 (0.0012) & 0.024 (0.0024) & 1.000 (0.0000) & 0.172 (0.0096) & 0.444 (0.0112) & 0.762 (0.0091) & 0.09 (0.024) \\
 & Horseshoe & 0.008 (0.0005) & 0.015 (0.0017) & 1.000 (0.0000) & 0.058 (0.0098) & 0.688 (0.0321) & 0.492 (0.0412) & 1.71 (0.022) \\
 & Horseshoe+ & \textbf{0.007 (0.0004)} & 0.012 (0.0013) & 1.000 (0.0000) & 0.041 (0.0048) & 0.741 (0.0200) & 0.426 (0.0280) & 1.79 (0.030) \\
 & Mhorseshoe & \textbf{0.007 (0.0005)} & \textbf{0.010 (0.0012)} & 1.000 (0.0000) & 0.018 (0.0037) & 0.864 (0.0243) & 0.235 (0.0398) & 17.00 (1.847) \\
 & BUGS & 0.008 (0.0006) & \textbf{0.010 (0.0018)} & 1.000 (0.0000) & 0.002 (0.0008) & \textbf{0.985 (0.0075)} & \textbf{0.027 (0.0139)} & 48.43 (0.153) \\ \hline
\multirow{10}{*}{\begin{tabular}[c]{@{}c@{}}$p =500$ \\ $n =150$\end{tabular}} & LASSO & 0.007 (0.0005) & 0.017 (0.0018) & 1.000 (0.0000) & 0.031 (0.0043) & 0.639 (0.0296) & 0.572 (0.0381) & 0.09 (0.005) \\
 & UniLASSO & 0.006 (0.0005) & 0.021 (0.0029) & 0.990 (0.0100) & 0.014 (0.0026) & 0.779 (0.0328) & 0.369 (0.0524) & 0.07 (0.005) \\
 & BayesLASSO & 0.016 (0.0007) & 0.044 (0.0020) & 1.000 (0.0000) & 0.259 (0.0102) & 0.234 (0.0062) & 0.926 (0.0029) & 8.84 (0.574) \\
 & Dirich-Laplace & 0.006 (0.0011) & 0.017 (0.0072) & 1.000 (0.0000) & 0.000 (0.0003) & \textbf{0.990 (0.0063)} & \textbf{0.018 (0.0121)} & 170.94 (17.476) \\
 & R2D2 & \textbf{0.004 (0.0002)} & 0.015 (0.0011) & 1.000 (0.0000) & 0.014 (0.0017) & 0.777 (0.0235) & 0.384 (0.0381) & 165.82 (16.126) \\
 & SSLASSO & 0.009 (0.0004) & 0.031 (0.0020) & 1.000 (0.0000) & 0.094 (0.0033) & 0.404 (0.0072) & 0.820 (0.0059) & 0.32 (0.040) \\
 & Horseshoe & 0.005 (0.0005) & 0.020 (0.0042) & 1.000 (0.0000) & 0.025 (0.0075) & 0.715 (0.0478) & 0.459 (0.0631) & 9.78 (0.125) \\
 & Horseshoe+ & \textbf{0.004 (0.0003)} & \textbf{0.013 (0.0018)} & 1.000 (0.0000) & 0.014 (0.0030) & 0.780 (0.0326) & 0.375 (0.0495) & 9.81 (0.033) \\
 & Mhorseshoe & \textbf{0.003 (0.0002)} & \textbf{0.008 (0.0008)} & 1.000 (0.0000) & 0.004 (0.0011) & 0.910 (0.0211) & 0.165 (0.0376) & 66.39 (6.414) \\
 & BUGS & \textbf{0.004 (0.0015)} & 0.020 (0.0158) & 1.000 (0.0000) & 0.000 (0.0002) & \textbf{0.995 (0.0048)} & \textbf{0.009 (0.0091)} & 120.77 (0.468) \\ \hline
\multirow{10}{*}{\begin{tabular}[c]{@{}c@{}}$p =1,000$ \\ $n =200$\end{tabular}} & LASSO & 0.005 (0.0002) & 0.015 (0.0013) & 1.000 (0.0000) & 0.014 (0.0026) & 0.665 (0.0353) & 0.541 (0.0470) & 0.15 (0.011) \\
 & UniLASSO & 0.004 (0.0005) & 0.019 (0.0035) & 0.970 (0.0153) & 0.008 (0.0011) & 0.737 (0.0251) & 0.432 (0.0342) & 0.12 (0.003) \\
 & BayesLASSO & 0.014 (0.0006) & 0.044 (0.0017) & 1.000 (0.0000) & 0.113 (0.0078) & 0.273 (0.0089) & 0.915 (0.0047) & 14.02 (0.310) \\
 & Dirich-Laplace & 0.004 (0.0013) & 0.030 (0.0166) & 1.000 (0.0000) & 0.000 (0.0000) & \textbf{1.000 (0.0000)} & \textbf{0.000 (0.0000)} & 895.74 (7.683) \\
 & R2D2 & 0.003 (0.0001) & 0.015 (0.0010) & 1.000 (0.0000) & 0.004 (0.0011) & 0.845 (0.0287) & 0.275 (0.0475) & 894.96 (5.556) \\
 & SSLASSO & 0.006 (0.0002) & 0.032 (0.0014) & 1.000 (0.0000) & 0.055 (0.0023) & 0.386 (0.0069) & 0.842 (0.0052) & 0.91 (0.094) \\
 & Horseshoe & 0.003 (0.0004) & 0.017 (0.0048) & 1.000 (0.0000) & 0.009 (0.0032) & 0.791 (0.0532) & 0.345 (0.0781) & 14.81 (0.268) \\
 & Horseshoe+ & 0.004 (0.0002) & 0.024 (0.0028) & 1.000 (0.0000) & 0.021 (0.0029) & 0.584 (0.0309) & 0.644 (0.0372) & 15.35 (0.246) \\
 & Mhorseshoe & \textbf{0.002 (0.0001)} & \textbf{0.006 (0.0005)} & 1.000 (0.0000) & 0.002 (0.0005) & 0.938 (0.0210) & 0.114 (0.0382) & 199.68 (16.679) \\
 & BUGS & \textbf{0.002 (0.0001)} & \textbf{0.003 (0.0003)} & 1.000 (0.0000) & 0.000 (0.0001) & \textbf{0.995 (0.0047)} & \textbf{0.009 (0.0091)} & 292.87 (0.801) \\ \hline
\multirow{11}{*}{\begin{tabular}[c]{@{}c@{}}$p =10,000$ \\ $n =200$\end{tabular}} & LASSO & 0.002 (0.0002) & 0.028 (0.0020) & 1.000 (0.0000) & 0.002 (0.0003) & 0.602 (0.0303) & 0.629 (0.0369) & 1.27 (0.015) \\
 & UniLASSO & 0.002 (0.0003) & 0.030 (0.0066) & 0.970 (0.0153) & 0.002 (0.0006) & 0.577 (0.0465) & 0.638 (0.0537) & 2.60 (0.169) \\
 & BayesLASSO & 0.010 (0.0002) & 0.050 (0.0022) & 0.240 (0.0163) & 0.000 (0.0000) & 0.487 (0.0164) & \textbf{0.000 (0.0000)} & 96.00 (16.883) \\
 & SSLASSO & 0.002 (0.0001) & 0.044 (0.0021) & 1.000 (0.0000) & 0.005 (0.0003) & 0.414 (0.0085) & 0.827 (0.0068) & 38.72 (4.245) \\
 & Horseshoe & \textbf{0.001 (0.0000)} & 0.051 (0.0021) & 1.000 (0.0000) & 0.001 (0.0001) & 0.726 (0.0187) & 0.470 (0.0280) & 113.87 (19.247) \\
 & Horseshoe+ & \textbf{0.001 (0.0000)} & 0.051 (0.0021) & 1.000 (0.0000) & 0.001 (0.0001) & 0.667 (0.0132) & 0.553 (0.0176) & 115.01 (20.001) \\
 & SIS-LASSO & 0.005 (0.0004) & 0.091 (0.0131) & 0.650 (0.0342) & 0.002 (0.0001) & 0.377 (0.0208) & 0.780 (0.0128) & 0.04 (0.003) \\
 & SIS-SCAD & 0.004 (0.0005) & 0.093 (0.0133) & 0.650 (0.0342) & 0.002 (0.0003) & 0.450 (0.0475) & 0.677 (0.0523) & 0.05 (0.005) \\
 & SIS-RHS & 0.004 (0.0004) & 0.090 (0.0136) & 0.650 (0.0342) & 0.002 (0.0001) & 0.389 (0.0269) & 0.765 (0.0206) & 11.33 (0.052) \\
 & BUGS & \textbf{0.001 (0.0001)} & \textbf{0.006 (0.0011)} & 1.000 (0.0000) & 0.000 (0.0000) & \textbf{0.991 (0.0062)} & \textbf{0.018 (0.0121)} & 4435.43 (21.596) \\
 & BUGS-Active & \textbf{0.001 (0.0001)} & \textbf{0.008 (0.0006)} & 1.000 (0.0000) & 0.000 (0.0000) & \textbf{0.987 (0.0094)} & 0.026 (0.0181) & 2399.43 (7.660) \\ \hline
\multirow{6}{*}{\begin{tabular}[c]{@{}c@{}}$p =100,000$ \\ $n =200$\end{tabular}} & LASSO & 0.0009 (0.00007) & 0.0328 (0.00259) & 0.9800 (0.01333) & 0.0003 (0.00003) & 0.5045 (0.02173) & 0.7366 (0.02091) & 18.2 (1.88) \\
 & UniLASSO & 0.0011 (0.00011) & 0.0465 (0.00513) & 0.9000 (0.02981) & 0.0006 (0.00004) & 0.3611 (0.01731) & 0.8540 (0.01000) & 19.5 (0.66) \\
 & SIS-LASSO & 0.0019 (0.00013) & 0.1145 (0.01043) & 0.5200 (0.03266) & 0.0003 (0.00001) & 0.2988 (0.02256) & 0.8276 (0.01567) & 0.2  (0.02) \\
 & SIS-SCAD & 0.0019 (0.00013) & 0.1167 (0.01046) & 0.5200 (0.03266) & 0.0002 (0.00002) & 0.3090 (0.02440) & 0.8134 (0.02074) & 0.1 (0.02) \\
 & SIS-RHS & 0.0018 (0.00013) & 0.1132 (0.01033) & 0.5200 (0.03266) & 0.0003 (0.00001) & 0.2979 (0.02277) & 0.8286 (0.01578) & 9.5 (0.07) \\
 & BUGS-Active & \textbf{0.0008 (0.00010)} & \textbf{0.0311 (0.00500)} & 0.8700 (0.02603) & 0.0000 (0.00000) & \textbf{0.9218 (0.01610)} & \textbf{0.0211 (0.01410)} & 5580.4 (14.58) \\ \hline
\multirow{4}{*}{\begin{tabular}[c]{@{}c@{}}$p =1,000,000$ \\ $n =500$\end{tabular}} & SIS-LASSO & 0.0004 (0.00004) & 0.0652 (0.00854) & 0.7200 (0.03590) & 0.0000 (0.00000) & 0.3288 (0.02327) & 0.8485 (0.01450) & 2 (0.0) \\
 & SIS-SCAD & 0.0004 (0.00004) & 0.0679 (0.00850) & 0.7200 (0.03590) & 0.0000 (0.00001) & 0.3701 (0.03923) & 0.7942 (0.04369) & 2 (0.0) \\
 & SIS-RHS & \textbf{0.0003 (0.00003)} & \textbf{0.0634 (0.00874)} & 0.7200 (0.03590) & 0.0000 (0.00000) & 0.3898 (0.03839) & 0.7817 (0.03421) & 58 (0.5) \\
 & BUGS-Active & \textbf{0.0003 (0.00002)} & 0.0914 (0.01032) & 0.6500 (0.03727) & 0.0000 (0.00000) & \textbf{0.8034 (0.02230)} & \textbf{0.0000 (0.00000)} & 73575 (169.5) \\ \hline
\end{tabular}}
\caption{Simulation results for Scenario 1 comparing LASSO, UniLASSO, Bayesian LASSO, Dirichlet--Laplace, R2D2, spike-and-slab LASSO (SSLASSO), Horseshoe, Horseshoe+, Mhorseshoe, SIS-LASSO, SIS-SCAD, SIS-RHS, BUGS, and BUGS-Active. Performance is evaluated using root mean squared error of the regression coefficients (RMSE ($\beta$)), mean squared prediction error (MSE ($y$)), true positive rate (TPR), false positive rate (FPR), Matthews correlation coefficient (MCC), false discovery rate (FDR), and computational time in seconds (CompTime). Values are reported as mean (standard error) over 10 independent replications. For dimensions up to $p=10^4$, the two best-performing methods with respect to RMSE ($\beta$), MSE ($y$), MCC, and FDR are highlighted in bold; for $p\geq10^5$, the best-performing method under each metric is highlighted, with ties retained.}\vspace{-0.5cm}
\label{tab:scenario_1}
\end{table}

For Scenario 1, Table~\ref{tab:scenario_1} shows that BUGS consistently achieves an advantageous balance between sensitivity and specificity across moderate and high-dimensional regimes. While most competing methods attain near-perfect TPR, they do so at the cost of substantially higher false discoveries, leading to inferior MCC. In contrast, BUGS maintains essentially perfect signal recovery while achieving markedly lower FDR, resulting in high MCC. Among competitors, Dirichlet--Laplace and Mhorseshoe provide the closest performance. Dirichlet--Laplace yields reduced empirical FDR and competitive MCC, but at the expense of estimation and prediction accuracy, reflecting overly conservative shrinkage. Mhorseshoe, on the other hand, performs strongly in estimation and prediction, but exhibits higher FDR, leading to lower MCC relative to BUGS. Overall, BUGS offers a more balanced operating point, combining competitive estimation and prediction with substantially reduced false discoveries. The advantage in variable-selection accuracy remains pronounced as dimensionality increases. At $p=10^4$, BUGS and BUGS-Active substantially outperform the SIS-based competitors in estimation, prediction, and support recovery, while BUGS-Active reduces computational cost relative to BUGS. At $p=10^5$, BUGS-Active similarly achieves markedly higher MCC and lower FDR than SIS-LASSO, SIS-SCAD, and SIS-RHS while remaining competitive in estimation and prediction. At $p=10^6$, the comparison reveals a clearer trade-off: the SIS-based methods achieve lower prediction error and somewhat higher TPR, but incur very high FDR and substantially lower MCC. In contrast, BUGS-Active yields much stronger support recovery, with an MCC of $0.803$ and essentially no false discoveries, although at considerably greater computational cost. These results highlight the distinction between predictive accuracy and reliable sparse signal recovery in ultra-high-dimensional settings.

\begin{table}[!t]
\renewcommand{\arraystretch}{1.15}
\resizebox{0.99\columnwidth}{!}{%
\begin{tabular}{c|l|ccccccc}
\hline
$(p,n)$ & Method & RMSE ($\beta$) & MSE ($y$) & TPR & FPR & MCC & FDR & \begin{tabular}[c]{@{}c@{}}CompTime\\ (sec)\end{tabular} \\ \hline
\multirow{10}{*}{\begin{tabular}[c]{@{}c@{}}$p =200$ \\ $n =100$\end{tabular}} & LASSO & 0.032 (0.0033) & 0.051 (0.0054) & 0.990 (0.0100) & 0.173 (0.0152) & 0.443 (0.0226) & 0.757 (0.0182) & 0.06 (0.005) \\
 & UniLASSO & 0.071 (0.0058) & 0.334 (0.0418) & 0.450 (0.0500) & 0.026 (0.0082) & 0.465 (0.0319) & 0.417 (0.0793) & 0.07 (0.009) \\
 & BayesLASSO & 0.055 (0.0042) & 0.074 (0.0073) & 0.990 (0.0100) & 0.638 (0.0109) & 0.162 (0.0065) & 0.924 (0.0015) & 2.06 (0.147) \\
 & Dirich-Laplace & 0.052 (0.0069) & 0.196 (0.0396) & 0.790 (0.0526) & 0.004 (0.0016) & \textbf{0.842 (0.0405)} & \textbf{0.080 (0.0339)} & 36.04 (3.347) \\
 & R2D2 & \textbf{0.015 (0.0015)} & \textbf{0.025 (0.0033)} & 1.000 (0.0000) & 0.065 (0.0095) & 0.663 (0.0292) & 0.526 (0.0370) & 36.53 (2.967) \\
 & SSLASSO & 0.037 (0.0044) & 0.058 (0.0062) & 0.990 (0.0100) & 0.261 (0.0139) & 0.351 (0.0126) & 0.830 (0.0079) & 0.22 (0.045) \\
 & Horseshoe & 0.016 (0.0014) & 0.028 (0.0036) & 1.000 (0.0000) & 0.089 (0.0142) & 0.604 (0.0375) & 0.591 (0.0459) & 1.96 (0.025) \\
 & Horseshoe+ & \textbf{0.014 (0.0013)} & \textbf{0.024 (0.0032)} & 1.000 (0.0000) & 0.054 (0.0086) & 0.703 (0.0348) & 0.470 (0.0487) & 2.03 (0.022) \\
 & Mhorseshoe & 0.016 (0.0015) & \textbf{0.025 (0.0033)} & 1.000 (0.0000) & 0.049 (0.0115) & 0.731 (0.0415) & 0.428 (0.0579) & 14.56 (1.015) \\
 & BUGS & 0.022 (0.0020) & 0.035 (0.0056) & 0.990 (0.0100) & 0.014 (0.0042) & \textbf{0.886 (0.0278)} & \textbf{0.189 (0.0459)} & 61.38 (0.147) \\ \hline
\multirow{10}{*}{\begin{tabular}[c]{@{}c@{}}$p =500$ \\ $n =150$\end{tabular}} & LASSO & 0.019 (0.0013) & 0.052 (0.0045) & 0.990 (0.0100) & 0.097 (0.0092) & 0.404 (0.0226) & 0.815 (0.0184) & 0.08 (0.005) \\
 & UniLASSO & 0.040 (0.0025) & 0.263 (0.0333) & 0.600 (0.0298) & 0.025 (0.0066) & 0.479 (0.0397) & 0.551 (0.0816) & 0.09 (0.004) \\
 & BayesLASSO & 0.039 (0.0017) & 0.077 (0.0051) & 0.970 (0.0153) & 0.435 (0.0104) & 0.151 (0.0041) & 0.956 (0.0009) & 3.99 (0.070) \\
 & Dirich-Laplace & 0.021 (0.0029) & 0.106 (0.0394) & 0.900 (0.0365) & 0.001 (0.0005) & \textbf{0.926 (0.0242)} & \textbf{0.040 (0.0214)} & 362.15 (22.570) \\
 & R2D2 & 0.008 (0.0007) & 0.028 (0.0026) & 1.000 (0.0000) & 0.027 (0.0030) & 0.653 (0.0210) & 0.558 (0.0268) & 349.64 (24.327) \\
 & SSLASSO & 0.020 (0.0015) & 0.067 (0.0042) & 0.970 (0.0153) & 0.149 (0.0057) & 0.311 (0.0109) & 0.881 (0.0051) & 0.81 (0.113) \\
 & Horseshoe & 0.008 (0.0009) & 0.031 (0.0061) & 1.000 (0.0000) & 0.033 (0.0091) & 0.655 (0.0435) & 0.543 (0.0533) & 4.15 (0.063) \\
 & Horseshoe+ & \textbf{0.007 (0.0007)} & 0.023 (0.0032) & 1.000 (0.0000) & 0.022 (0.0032) & 0.699 (0.0271) & 0.495 (0.0374) & 4.43 (0.044) \\
 & Mhorseshoe & \textbf{0.007 (0.0008)} & \textbf{0.020 (0.0023)} & 1.000 (0.0000) & 0.013 (0.0025) & 0.796 (0.0315) & 0.350 (0.0500) & 67.54 (5.155) \\
 & BUGS & 0.008 (0.0011) & \textbf{0.017 (0.0027)} & 0.980 (0.0200) & 0.003 (0.0010) & \textbf{0.931 (0.0199)} & \textbf{0.108 (0.0376)} & 146.41 (0.318) \\ \hline
\multirow{10}{*}{\begin{tabular}[c]{@{}c@{}}$p =1,000$ \\ $n =200$\end{tabular}} & LASSO & 0.012 (0.0009) & 0.044 (0.0040) & 1.000 (0.0000) & 0.050 (0.0057) & 0.413 (0.0191) & 0.818 (0.0156) & 0.13 (0.005) \\
 & UniLASSO & 0.029 (0.0022) & 0.268 (0.0396) & 0.600 (0.0596) & 0.013 (0.0045) & 0.493 (0.0314) & 0.509 (0.0859) & 0.17 (0.007) \\
 & BayesLASSO & 0.032 (0.0009) & 0.075 (0.0030) & 0.920 (0.0249) & 0.255 (0.0092) & 0.151 (0.0069) & 0.964 (0.0016) & 14.46 (1.586) \\
 & Dirich-Laplace & 0.014 (0.0020) & 0.077 (0.0214) & 0.980 (0.0200) & 0.000 (0.0001) & \textbf{0.984 (0.0158)} & \textbf{0.011 (0.0111)} & 890.62 (8.651) \\
 & R2D2 & 0.004 (0.0003) & 0.026 (0.0014) & 1.000 (0.0000) & 0.010 (0.0012) & 0.708 (0.0210) & 0.490 (0.0294) & 895.40 (6.293) \\
 & SSLASSO & 0.016 (0.0021) & 0.070 (0.0027) & 1.000 (0.0000) & 0.099 (0.0045) & 0.291 (0.0062) & 0.906 (0.0035) & 3.13 (0.374) \\
 & Horseshoe & 0.004 (0.0003) & 0.024 (0.0058) & 1.000 (0.0000) & 0.011 (0.0038) & 0.732 (0.0416) & 0.444 (0.0540) & 14.79 (1.055) \\
 & Horseshoe+ & 0.006 (0.0007) & 0.042 (0.0069) & 1.000 (0.0000) & 0.029 (0.0051) & 0.529 (0.0369) & 0.701 (0.0396) & 14.98 (0.865) \\
 & Mhorseshoe & \textbf{0.003 (0.0002)} & \textbf{0.012 (0.0010)} & 1.000 (0.0000) & 0.003 (0.0006) & 0.883 (0.0219) & 0.214 (0.0389) & 129.63 (2.658) \\
 & BUGS & \textbf{0.003 (0.0003)} & \textbf{0.008 (0.0006)} & 1.000 (0.0000) & 0.001 (0.0003) & \textbf{0.973 (0.0129)} & \textbf{0.050 (0.0242)} & 348.02 (0.453) \\ \hline
\multirow{11}{*}{\begin{tabular}[c]{@{}c@{}}$p =10,000$ \\ $n =200$\end{tabular}} & LASSO & 0.009 (0.0004) & 0.116 (0.0116) & 0.780 (0.0200) & 0.005 (0.0005) & 0.338 (0.0195) & 0.849 (0.0161) & 2.95 (0.127) \\
 & UniLASSO & 0.011 (0.0003) & 0.187 (0.0280) & 0.520 (0.0200) & 0.007 (0.0012) & 0.251 (0.0531) & 0.824 (0.0928) & 1.17 (0.028) \\
 & BayesLASSO & 0.014 (0.0002) & 0.093 (0.0041) & 0.190 (0.0314) & 0.000 (0.0000) & 0.422 (0.0365) & \textbf{0.000 (0.0000)} & 67.48 (0.381) \\
 & SSLASSO & 0.006 (0.0006) & 0.089 (0.0037) & 0.920 (0.0389) & 0.010 (0.0004) & 0.281 (0.0144) & 0.913 (0.0054) & 182.74 (17.869) \\
 & Horseshoe & \textbf{0.002 (0.0001)} & 0.096 (0.0038) & 1.000 (0.0000) & 0.002 (0.0001) & 0.626 (0.0082) & 0.607 (0.0103) & 68.28 (0.379) \\
 & Horseshoe+ & \textbf{0.002 (0.0001)} & 0.096 (0.0038) & 1.000 (0.0000) & 0.002 (0.0001) & 0.575 (0.0096) & 0.668 (0.0110) & 69.74 (0.232) \\
 & SIS-LASSO & 0.012 (0.0004) & 0.244 (0.0085) & 0.520 (0.0249) & 0.003 (0.0001) & 0.283 (0.0145) & 0.844 (0.0086) & 0.04 (0.007) \\
 & SIS-SCAD & 0.012 (0.0004) & 0.248 (0.0080) & 0.520 (0.0249) & 0.003 (0.0001) & 0.287 (0.0168) & 0.840 (0.0121) & 0.07 (0.007) \\
 & SIS-RHS & 0.012 (0.0003) & 0.244 (0.0087) & 0.530 (0.0260) & 0.003 (0.0001) & 0.282 (0.0149) & 0.848 (0.0086) & 13.74 (0.039) \\
 & BUGS & 0.003 (0.0006) & \textbf{0.042 (0.0098)} & 0.960 (0.0267) & 0.000 (0.0001) & \textbf{0.903 (0.0248)} & 0.143 (0.0429) & 2155.98 (5.007) \\
 & BUGS-Active & \textbf{0.002 (0.0003)} & \textbf{0.022 (0.0029)} & 0.980 (0.0200) & 0.000 (0.0000) & \textbf{0.954 (0.0184)} & \textbf{0.069 (0.0263)} & 1596.66 (6.632) \\ \hline
\multirow{6}{*}{\begin{tabular}[c]{@{}c@{}}$p =100,000$ \\ $n =200$\end{tabular}} & LASSO & 0.0026 (0.00008) & 0.1651 (0.02497) & 0.5200 (0.02494) & 0.0003 (0.00008) & 0.3346 (0.03179) & 0.7591 (0.04565) & 12.3 (0.12) \\
 & UniLASSO & 0.0026 (0.00009) & \textbf{0.0776 (0.00968)} & 0.5200 (0.02906) & 0.0007 (0.00006) & 0.1950 (0.01276) & 0.9236 (0.01052) & 14.2 (0.29) \\
 & SIS-LASSO & 0.0030 (0.00008) & 0.1528 (0.00518) & 0.3800 (0.02906) & 0.0003 (0.00001) & 0.2092 (0.01754) & 0.8846 (0.01062) & 0.1 (0.00) \\
 & SIS-SCAD & 0.0030 (0.00007) & 0.1556 (0.00508) & 0.3800 (0.02906) & 0.0003 (0.00001) & 0.2062 (0.01674) & 0.8879 (0.00968) & 0.1 (0.00) \\
 & SIS-RHS & 0.0028 (0.00007) & 0.1496 (0.00524) & 0.3800 (0.02906) & 0.0003 (0.00001) & 0.2078 (0.01686) & 0.8861 (0.00991) & 11.7 (0.30) \\
 & BUGS-Active & \textbf{0.0023 (0.00009)} & 0.0804 (0.00474) & 0.4900 (0.02769) & 0.0000 (0.00001) & \textbf{0.5176 (0.03610)} & \textbf{0.4495 (0.05004)} & 5637.9 (12.32) \\ \hline
\multirow{4}{*}{\begin{tabular}[c]{@{}c@{}}$p =1,000,000$ \\ $n =500$\end{tabular}} & SIS-LASSO & 0.0009 (0.00002) & 0.1410 (0.00555) & 0.5300 (0.03000) & 0.0001 (0.00000) & 0.2038 (0.01159) & 0.9216 (0.00455) & 3 (0.2) \\
 & SIS-SCAD & 0.0009 (0.00002) & 0.1448 (0.00542) & 0.5300 (0.03000) & 0.0001 (0.00000) & 0.2003 (0.01147) & 0.9243 (0.00445) & 4 (0.2) \\
 & SIS-RHS & \textbf{0.0008 (0.00002)} & \textbf{0.1385 (0.00566)} & 0.5300 (0.03000) & 0.0001 (0.00000) & 0.2057 (0.01164) & 0.9201 (0.00460) & 42 (0.1) \\
 & BUGS-Active & \textbf{0.0008 (0.00001)} & 0.2077 (0.00623) & 0.4500 (0.01667) & 0.0000 (0.00000) & \textbf{0.6631 (0.01596)} & \textbf{0.0200 (0.02000)} & 73291 (57.9) \\ \hline
\end{tabular}}
\caption{Simulation results for Scenario 2 comparing LASSO, UniLASSO, Bayesian LASSO, Dirichlet--Laplace, R2D2, spike-and-slab LASSO (SSLASSO), Horseshoe, Horseshoe+, Mhorseshoe, SIS-LASSO, SIS-SCAD, SIS-RHS, BUGS, and BUGS-Active. Performance is evaluated using root mean squared error of the regression coefficients (RMSE ($\beta$)), mean squared prediction error (MSE ($y$)), true positive rate (TPR), false positive rate (FPR), Matthews correlation coefficient (MCC), false discovery rate (FDR), and computational time in seconds (CompTime). Values are reported as mean (standard error) over 10 independent replications. For dimensions up to $p=10^4$, the two best-performing methods with respect to RMSE ($\beta$), MSE ($y$), MCC, and FDR are highlighted in bold; for $p\geq10^5$, the best-performing method under each metric is highlighted, with ties retained.}\vspace{-0.5cm}
\label{tab:scenario_2}
\end{table}

For Scenario 2, the correlated design substantially increases the difficulty of variable selection, yet the overall pattern in Table~\ref{tab:scenario_2} continues to favor the proposed approach. Across moderate and high-dimensional regimes, BUGS maintains high or near-perfect TPR while achieving substantially lower FDR than most competing methods, resulting in high MCC. In contrast, methods such as LASSO, horseshoe variants, R2D2, and SSLASSO either incur markedly higher false discoveries or become overly conservative and lose signal recovery under correlation. Among competitors, Dirichlet--Laplace and Mhorseshoe again provide the closest alternatives. Dirichlet--Laplace achieves substantially reduced empirical FDR and competitive MCC, but at the cost of estimation and prediction accuracy due to conservative shrinkage. Mhorseshoe performs well in RMSE$(\beta)$ and MSE$(y)$, but exhibits higher FDR, leading to lower MCC relative to BUGS. Overall, BUGS maintains a more favorable balance between sensitivity and specificity, even under correlated designs. The advantage of the guided approach becomes more pronounced as dimensionality increases. At $p=10^4$, BUGS and BUGS-Active substantially outperform the SIS-based procedures in estimation, prediction, and support recovery, with markedly higher MCC and lower FDR; BUGS-Active closely matches or improves upon the full BUGS fit while reducing computational cost. At $p=10^5$, BUGS-Active continues to provide the strongest overall variable-selection performance, achieving substantially higher MCC and lower FDR than the SIS-based procedures while remaining competitive in estimation and prediction. At $p=10^6$, a clearer trade-off emerges: SIS-LASSO, SIS-SCAD, and SIS-RHS attain lower prediction error and somewhat higher TPR, but at the cost of very high FDR and substantially lower MCC. In contrast, BUGS-Active maintains a much cleaner selected model, with an MCC of $0.663$ and FDR of $0.020$, although at substantially greater computational cost.

Overall, these results demonstrate that the proposed guided shrinkage approach remains effective under correlation and continues to provide a favorable trade-off between sensitivity and specificity, while scaling to ultra-high-dimensional settings.

\noindent \textbf{Same-sample versus independent guidance.} Because BUGS constructs the marginal guidance from the same response used in the likelihood, we performed a targeted sensitivity analysis to assess whether this reuse materially affects variable-selection behavior. For each of the three moderate-dimensional regimes in Tables~\ref{tab:scenario_1}--\ref{tab:scenario_2}, we generated $2n_{\mathrm{fit}}$ observations and randomly divided them into two equal half-samples, denoted $A$ and $B$, so that each likelihood fit retained $n_{\mathrm{fit}}\in\{100,150,200\}$. Same-sample guidance (BUGS-SS), corresponding to the default BUGS implementation used throughout the paper, computed the guidance from the same half-sample used for the likelihood. In contrast, independent guidance (BUGS-IG) computed the guidance from the other half-sample: for a likelihood fit using $A$, guidance was constructed from $B$, and vice versa. Thus, BUGS-IG used disjoint half-samples for guidance construction and likelihood fitting, although response centering and scaling were performed
using the full sample before splitting. Within each replicate, agreement was assessed using two complementary
summaries. For the Jaccard index, posterior mean coefficient estimates
from the two likelihood fits were averaged separately for BUGS-SS and
BUGS-IG, and variables were selected when the absolute averaged estimate
exceeded $0.01$. For TPR and empirical FDR, variable selection and the
corresponding metrics were computed separately within each likelihood
fit using the same coefficient threshold. These metrics were then
averaged across the two fits for each method, and absolute paired
differences between BUGS-SS and BUGS-IG were calculated
(Table~\ref{tab:guidance_sensitivity}). Thus, the Jaccard index measures
agreement between selections based on averaged coefficient estimates,
whereas the TPR and FDR differences compare average selection performance
across the two likelihood fits.
\begin{table}[!h]
\renewcommand{\arraystretch}{1.2}
\resizebox{0.99\columnwidth}{!}{%
\begin{tabular}{l|ccc|ccc}
\hline
\multirow{2}{*}{Metric} & \multicolumn{3}{c|}{Scenario 1} & \multicolumn{3}{c}{Scenario 2} \\ \cline{2-7} 
 & \begin{tabular}[c]{@{}c@{}}$p = 200$\\ $n_{fit} = 100$\end{tabular} & \begin{tabular}[c]{@{}c@{}}$p = 500$\\ $n_{fit} = 150$\end{tabular} & \begin{tabular}[c]{@{}c@{}}$p = 1000$\\ $n_{fit} = 200$\end{tabular} & \begin{tabular}[c]{@{}c@{}}$p = 200$\\ $n_{fit} = 100$\end{tabular} & \begin{tabular}[c]{@{}c@{}}$p = 500$\\ $n_{fit} = 150$\end{tabular} & \begin{tabular}[c]{@{}c@{}}$p = 1000$\\ $n_{fit} = 200$\end{tabular} \\ \hline
Jaccard & 0.993 (0.004) & 0.998 (0.002) & 1.000 (0.000) & 0.945 (0.009) & 0.983 (0.005) & 0.980 (0.006) \\
$\lvert\Delta\mathrm{TPR}\rvert$ & 0.001 (0.001) & 0.000 (0.000) & 0.000 (0.000) & 0.014 (0.004) & 0.001 (0.001) & 0.001 (0.001) \\
$\lvert\Delta\mathrm{FDR}\rvert$ & 0.025 (0.005) & 0.008 (0.003) & 0.004 (0.002) & 0.057 (0.008) & 0.027 (0.005) & 0.031 (0.005) \\ \hline
\end{tabular}}
\caption{Sensitivity of BUGS to same-sample (SS) versus independent (IG) guidance. Here $n_{\mathrm{fit}}$ denotes the likelihood-sample size, matching the corresponding settings in Tables~\ref{tab:scenario_1}--\ref{tab:scenario_2}; the total generated sample size is $2n_{\mathrm{fit}}$. Entries are Monte Carlo means (standard errors). Jaccard indices are based on thresholded averaged coefficient estimates;
TPR and FDR differences are computed after averaging the corresponding metrics across the two likelihood fits within each replicate.}\vspace{-0.5cm}
\label{tab:guidance_sensitivity}
\end{table}

Across the settings considered, selections based on averaged coefficient estimates showed high agreement between BUGS-SS and BUGS-IG, with mean
Jaccard indices ranging from 0.945 to 1.000. Their average selection performance across the two likelihood fits was also similar, with mean
absolute differences no larger than 0.014 for TPR and 0.057 for empirical FDR. These findings suggest limited empirical sensitivity to the source of guidance for the summaries examined, although averaging can mask differences between individual half-sample fits.

\noindent \textbf{Sensitivity and convergence diagnostics.} Detailed results are reported in Supplementary Section S4. Sensitivity analysis across a range of hyperparameter perturbations (Supplementary Table~S1) shows stable performance relative to the baseline (Supplementary Table~S2), with negligible variation in estimation and prediction accuracy, consistently high TPR, and low FDR. Noticeable degradation occurs only under deliberately extreme settings, where increased flexibility in the guidance component leads to a moderate rise in false positives. Convergence diagnostics indicate reliable mixing. Gelman--Rubin statistics lie in the range \(0.999\)–\(1.021\) (Supplementary Table~S3), and effective sample sizes are in the hundreds for most parameters. As expected, the global scale parameter \(\tau\) mixes more slowly (ESS in the tens), but its posterior remains stable. Trace plots (Supplementary Figure~S1) show no evidence of non-stationarity, and posterior densities (Supplementary Figure~S2) are unimodal and well-separated between signal and noise coefficients.

%%%%%%%%%%%%%%%%%%%%%%%%%%%%%%%%%%%%%

\section{Developmental DNA Methylation Analysis in GUSTO}\label{sec:case_study}

We apply BUGS-Active to longitudinal DNA methylation data from the Growing Up in Singapore Towards healthy Outcomes (GUSTO) birth cohort \citep{Leroy2025}. Genome-wide methylation was measured using the Illumina Infinium MethylationEPIC BeadChip, yielding $p=865{,}859$ CpG probes, with $n=1051$ samples available for analysis. Measurements were obtained at ages $3$, $9$, $48$, and $72$ months, and age in months was treated as a continuous response. Our scientific objective is to identify a sparse set of CpGs associated with developmental age while retaining accurate prediction in the presence of nearly one million, substantially correlated methylation measurements. We therefore analyze all CpG probes jointly, without preliminary feature screening.

\noindent \textbf{Out-of-sample prediction and contribution of marginal guidance.}
We assess out-of-sample prediction using an 80--20 training--test split, with all preprocessing based on the training data only. BUGS-Active is compared with an unguided counterpart that retains the same hierarchical prior and computational scheme but removes the marginal guidance component. Posterior predictions are obtained using posterior mean coefficients, and predictive performance is evaluated by RMSE, MAE, Pearson correlation, and $R^2$. As shown in Table~\ref{tab:cv_8020}, BUGS-Active substantially improves upon its unguided counterpart across all four metrics, supporting the practical contribution of marginal guidance to out-of-sample prediction in this application.

\begin{table}[!h]
\centering
\resizebox{0.65\columnwidth}{!}{%
\begin{tabular}{lcccc}
\toprule
Method & RMSE & MAE & Corr & $R^2$ \\
\midrule
BUGS-Active (unguided) & 6.210 & 4.752 & 0.977 & 0.953 \\
BUGS-Active & 4.882 & 3.705 & 0.985 & 0.971 \\
\bottomrule
\end{tabular}}
\caption{Out-of-sample predictive performance under an 80--20 training--test split comparing BUGS-Active with and without marginal guidance. Corr denotes the Pearson correlation between observed and predicted age on the test set.}\vspace{-0.4cm}
\label{tab:cv_8020}
\end{table}

\noindent \textbf{Sparse identification and posterior characterization of age-associated CpGs.}
We next refit BUGS-Active to the full dataset to characterize CpGs most strongly associated with developmental age. Figure~\ref{fig:top10_cpgs_summary} summarizes posterior inference for the 10 CpGs with the largest absolute posterior mean effects, while Table~\ref{tab:top_cpgs} reports their genomic annotations and effect estimates on the original scale. The selected CpGs show clear posterior separation from zero, with all reported 95\% credible intervals excluding zero. Their posterior exceedance probabilities $P(|\beta_j|>0.01)$ are uniformly high, and the strongest signals remain stable across a range of selection thresholds. The selected sites span promoter-proximal regions, gene bodies, CpG islands, shores, shelves, and open-sea regions, indicating that the identified age-associated methylation signals are distributed across multiple genomic contexts.

\begin{figure}[!t]
\centering
\setlength{\figwidth}{\linewidth}

% -------- Row 1 --------
\begin{minipage}[t]{0.495\linewidth}
\centering
\includegraphics[width=\linewidth]{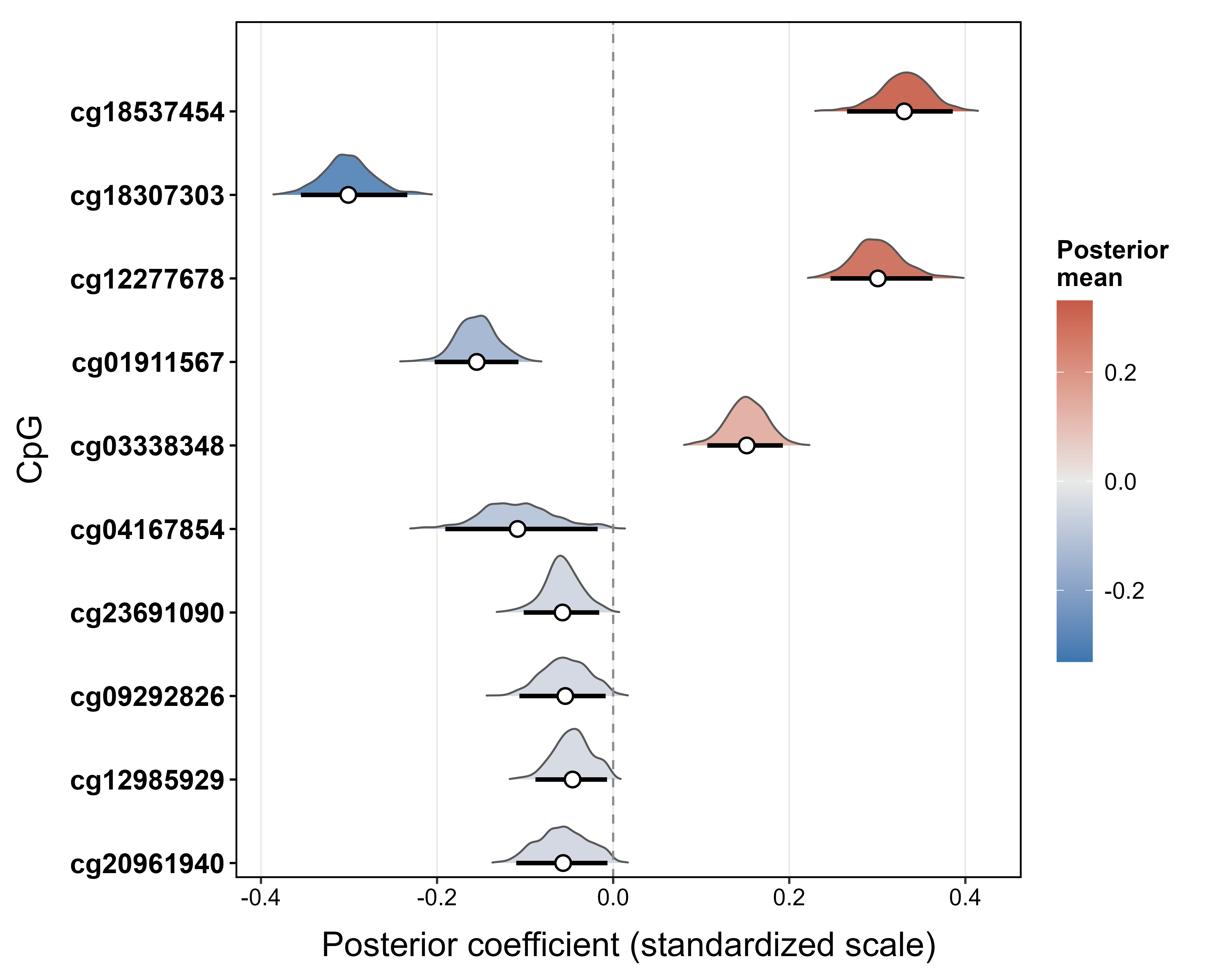}
\end{minipage}\hfill
\begin{minipage}[t]{0.495\linewidth}
\centering
\includegraphics[width=\linewidth]{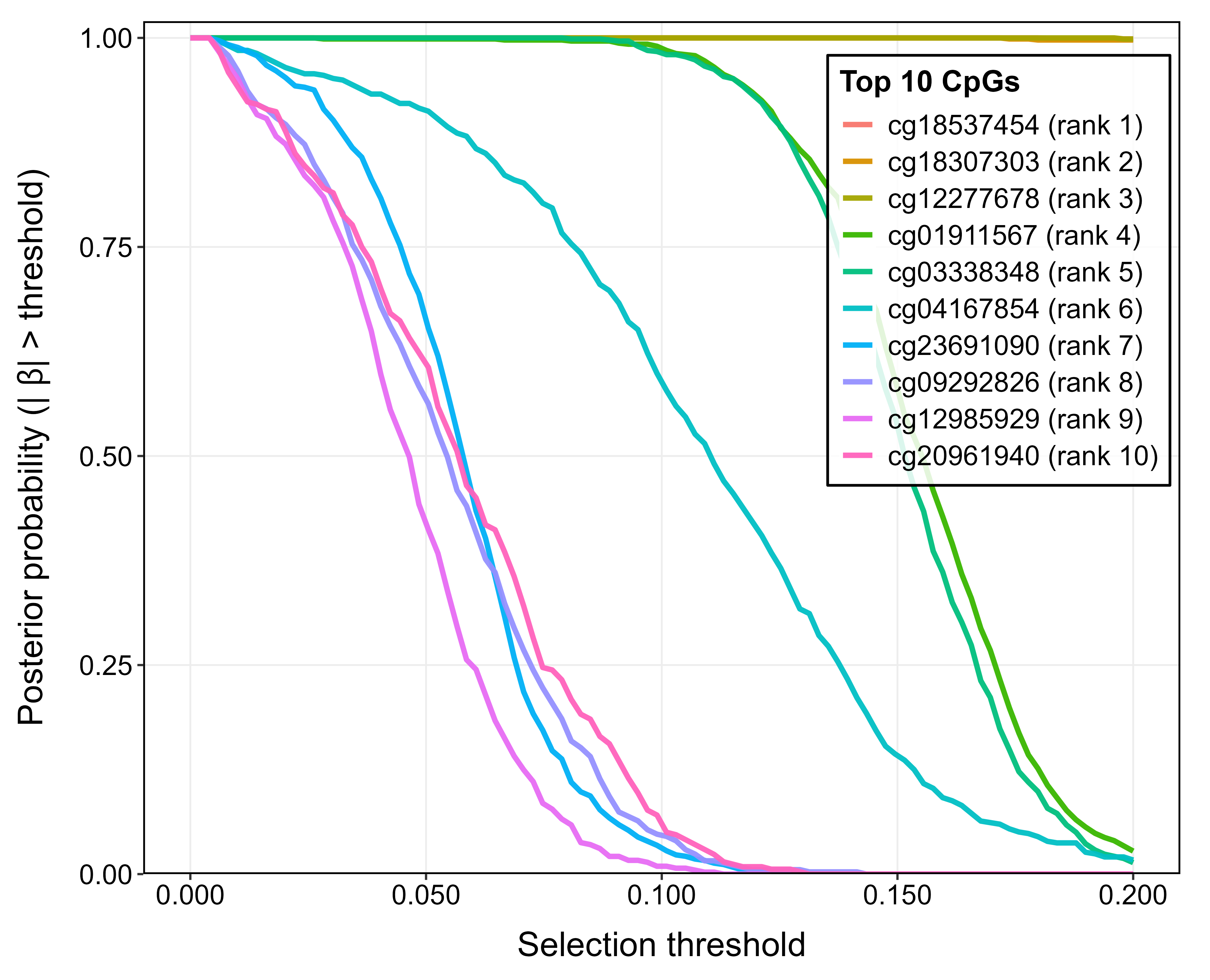}
\end{minipage}

\par\vspace{0.2cm}

% -------- Row 2 --------
\begin{minipage}[t]{0.495\linewidth}
\centering
\includegraphics[width=\linewidth]{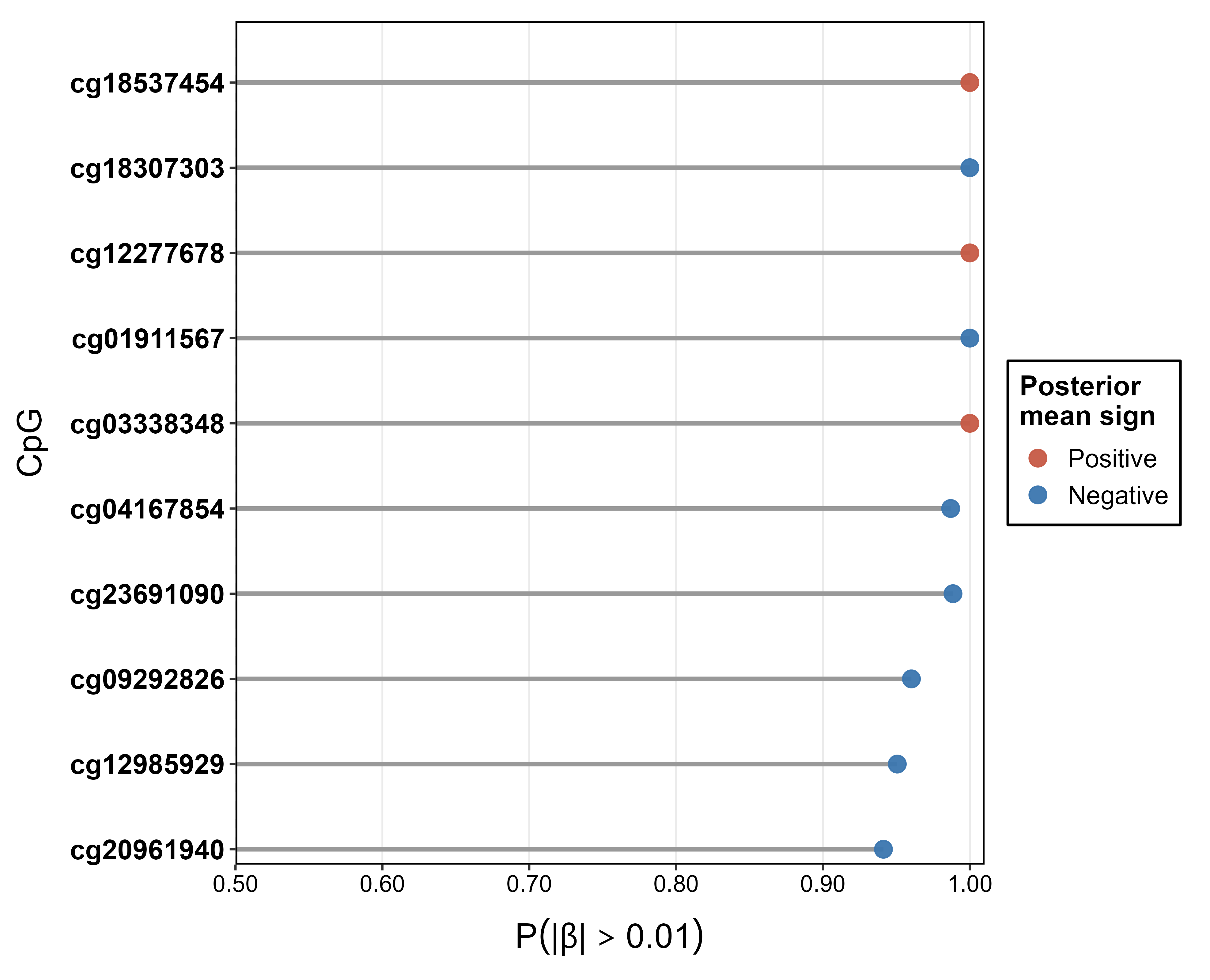}
\end{minipage}\hfill
\begin{minipage}[t]{0.495\linewidth}
\centering
\includegraphics[width=\linewidth]{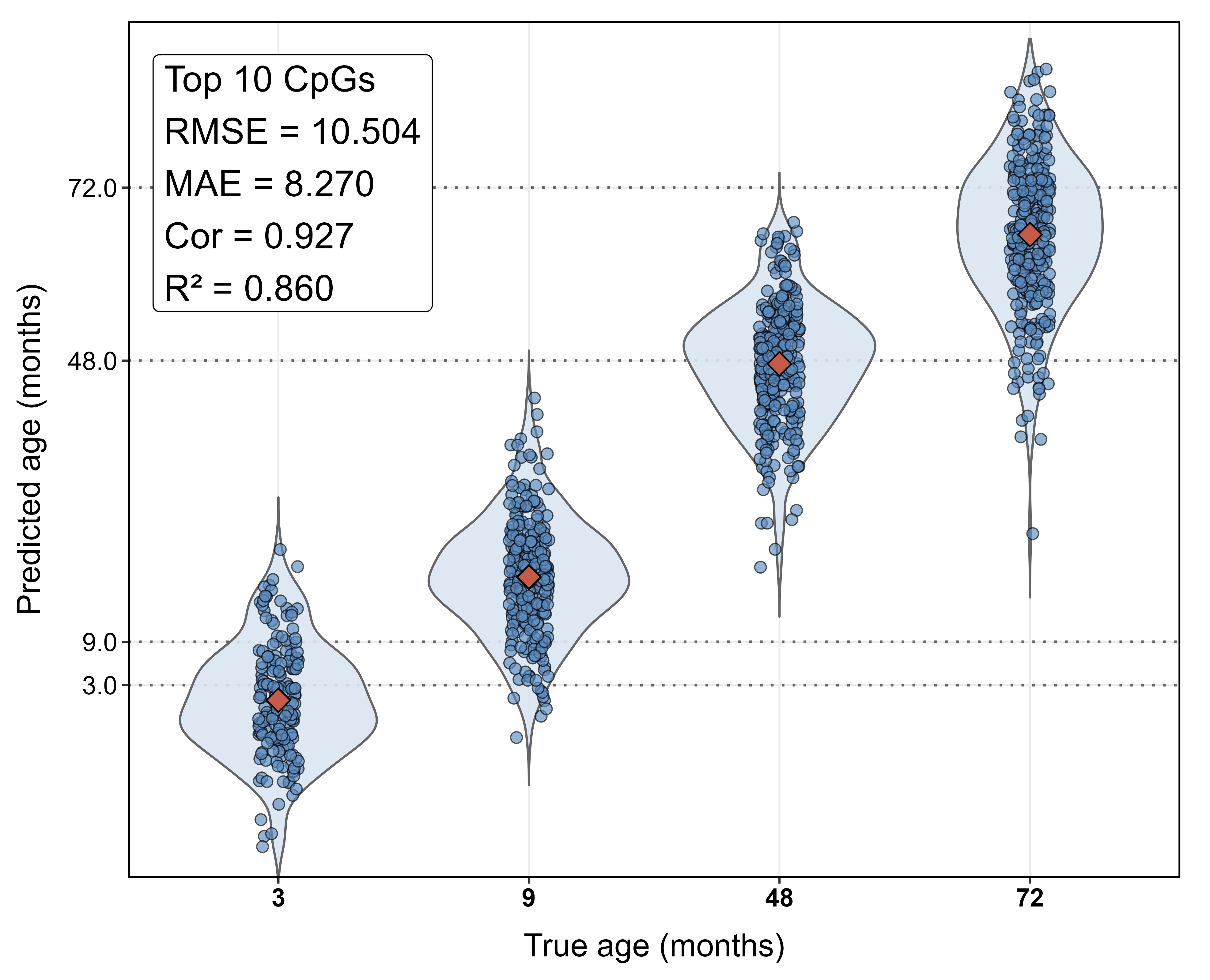}
\end{minipage}
\par

\caption{Posterior and predictive summaries for the top 10 CpGs
identified by BUGS-Active.
\textit{Top-left}: posterior coefficient distributions with posterior
means and credible intervals.
\textit{Top-right}: posterior exceedance probabilities across
selection thresholds.
\textit{Bottom-left}: posterior exceedance probabilities and effect
directions.
\textit{Bottom-right}: predicted versus observed age distributions
across developmental time points using the top 10 CpGs.}
\label{fig:top10_cpgs_summary}
\end{figure}

\begin{table}[!t]
\centering
\setlength{\tabcolsep}{4pt}
\rowcolors{2}{gray!8}{white}
\resizebox{0.8\columnwidth}{!}{%
\begin{tabular}{l l l l l r r}
\toprule
CpG & Gene(s) & Region & Chr:Pos & Island & $\hat{\beta}$ & 95\% CI \\
\midrule
\textbf{cg18537454} & \textemdash & \textemdash & chr10:22623213 & N\_Shore & 134.72 & [108.27, 157.19] \\
\shortstack{\textbf{cg18307303} \\ $\;$} & \shortstack{IL12B \\ $\;$} & \shortstack{1stExon/\\5'UTR} & \shortstack{chr5:158757456 \\ $\;$} & \shortstack{N\_Shore \\ $\;$} & \shortstack{-160.58 \\ $\;$} & \shortstack{[-189.31, -124.81] \\ $\;$} \\
\textbf{cg12277678} & LINC00703 & TSS200 & chr10:4426258 & OpenSea & 146.62 & [120.41, 176.92] \\
\textbf{cg01911567} & FRMD4A & Body & chr10:13727536 & OpenSea & -57.58 & [-75.36, -39.97] \\
\shortstack{\textbf{cg03338348} \\ $\;$ \\ $\;$} & \shortstack{TSNAXIP1,\\ RANBP10 \\ $\;$} & \shortstack{TSS1500/\\1stExon/\\5'UTR} & \shortstack{chr16:67840472 \\ $\;$ \\ $\;$} & \shortstack{Island \\ $\;$ \\ $\;$} & \shortstack{24.36 \\ $\;$ \\ $\;$} & \shortstack{[17.17, 30.95] \\ $\;$ \\ $\;$} \\
\textbf{cg04167854} & SERGEF & Body & chr11:18031127 & N\_Shelf & -30.28 & [-53.17, -4.93] \\
\shortstack{\textbf{cg23691090} \\ $\;$ \\ $\;$} & \shortstack{C22orf26,\\ LOC150381 \\ $\;$} & \shortstack{1stExon/\\5'UTR/\\Body} & \shortstack{chr22:46449981 \\ $\;$ \\ $\;$} & \shortstack{Island \\ $\;$ \\ $\;$} & \shortstack{-62.22 \\ $\;$ \\ $\;$} & \shortstack{[-109.90, -17.09] \\ $\;$ \\ $\;$} \\
\textbf{cg09292826} & \shortstack{TCTEX1D4,\\ BTBD19} & \shortstack{TSS1500/\\TSS200} & chr1:45274032 & S\_Shore & -38.77 & [-75.95, -6.10] \\
\shortstack{\textbf{cg12985929} \\ $\;$} & \shortstack{SEPT9 \\ $\;$} & \shortstack{5'UTR/\\Body} & \shortstack{chr17:75370611 \\ $\;$} & \shortstack{S\_Shore \\ $\;$} & \shortstack{-25.97 \\ $\;$} & \shortstack{[-49.70, -3.91] \\ $\;$} \\
\textbf{cg20961940} & SLC44A4 & Body & chr6:31832850 & S\_Shore & -36.79 & [-71.26, -4.18] \\
\bottomrule
\end{tabular}}
\vspace{2mm}
\caption{Top 10 CpG sites associated with developmental age identified by BUGS-Active. Gene(s) and Region give genomic annotations, Chr:Pos denotes chromosomal position, and Island denotes CpG island context. $\hat{\beta}$ is the posterior mean regression coefficient on the original scale.}\vspace{-0.4cm}
\label{tab:top_cpgs}
\end{table}

\noindent \textbf{Predictive information retained by the selected CpGs.}
To assess how much age-related information is retained by this sparse representation, we examine predictions based on the top 10 CpGs. These CpGs yield correlation $0.927$, $R^2=0.860$, RMSE $=10.504$, and MAE $=8.270$. Prediction is more closely aligned with observed age at $3$ and $48$ months and more variable at $9$ and $72$ months, indicating that predictive performance is not uniform across developmental stages.

Taken together, the GUSTO analysis illustrates the two practical aims of the proposed framework in an ultra-high-dimensional genomic setting: marginal guidance improves out-of-sample prediction relative to the matched unguided model, while joint analysis of all $865{,}859$ CpGs yields a compact set of age-associated sites with strong posterior support. Thus, BUGS-Active provides sparse, interpretable inference without requiring a preliminary screening step even when the predictor dimension approaches one million.

%%%%%%%%%%%%%%%%%%%%%%%%%%%%%%%%%%%%%%%%%%%%%%%
\section{Discussion}\label{sec:discussion}

The numerical and empirical results suggest that marginal guidance is most useful as a
\emph{soft source of information} within the shrinkage mechanism rather than as a hard
screening device. Across the simulation settings, BUGS generally retained high sensitivity
while reducing empirical false discoveries, with the benefit remaining evident under
correlated designs. The theoretical results provide a complementary interpretation: when
marginal guidance separates active from inactive predictors, it induces corresponding
shrinkage separation, whereas uninformative guidance does not itself create such separation.
The GUSTO analysis illustrates the same principle in an applied ultra-high-dimensional
setting, where incorporating marginal guidance improved out-of-sample prediction relative to
an otherwise matched unguided model while yielding a sparse set of CpGs with strong
posterior support. These findings indicate that marginal information can be useful even when
the final inferential target is fully multivariable, provided that it modifies rather than
determines variable inclusion.

An important practical limitation is that marginal evidence need not faithfully represent
conditional relevance, particularly in the presence of strong predictor dependence or signals
that are weak marginally but important jointly. BUGS partially mitigates this issue because
guidance affects the degree of shrinkage but does not exclude predictors from the full model.
The same consideration motivates the adaptive construction of BUGS-Active: variables with
large guidance scores are retained in the active set, while variables with weaker marginal
evidence may still enter dynamically when their posterior coefficient magnitudes become
sufficiently large. Thus, the coefficient-based component provides a posterior-driven
correction to purely marginal screening. Nevertheless, when marginal and conditional
evidence differ substantially, more iterations or a sufficiently flexible active set may be
needed for such variables to receive adequate local-scale adaptation.

For ultra-high-dimensional problems, BUGS-Active trades exact local-scale updating for
computational feasibility by restricting updates of $\{\lambda_j\}$ to a data-adaptive active
set while continuing to update regression coefficients and global parameters jointly. The
size of the \emph{guidance budget} is therefore an important practical tuning choice.
If the budget is too small, the approximation may become overly restrictive and distort the
global--local shrinkage dynamics; if it is unnecessarily large, much of the computational
advantage is lost. Figure~\ref{fig:guidance_effect} illustrates this trade-off under
Scenario~1 with $(p,n)=(1000,200)$ and $s_0=5$. A budget of $20$ variables yields visibly
degraded estimation despite containing all true signals, whereas increasing the budget to
$50$ produces posterior summaries close to those from full BUGS at substantially lower
computational cost ($8.13$ versus $65.23$ seconds). Increasing the budget further to $100$
provides little additional improvement, illustrating diminishing returns.

\begin{figure}[!t]
\centering
\setlength{\figwidth}{\linewidth}

% -------- Row 1 --------
\begin{minipage}[t]{0.495\linewidth}
\centering
\includegraphics[width=\linewidth]{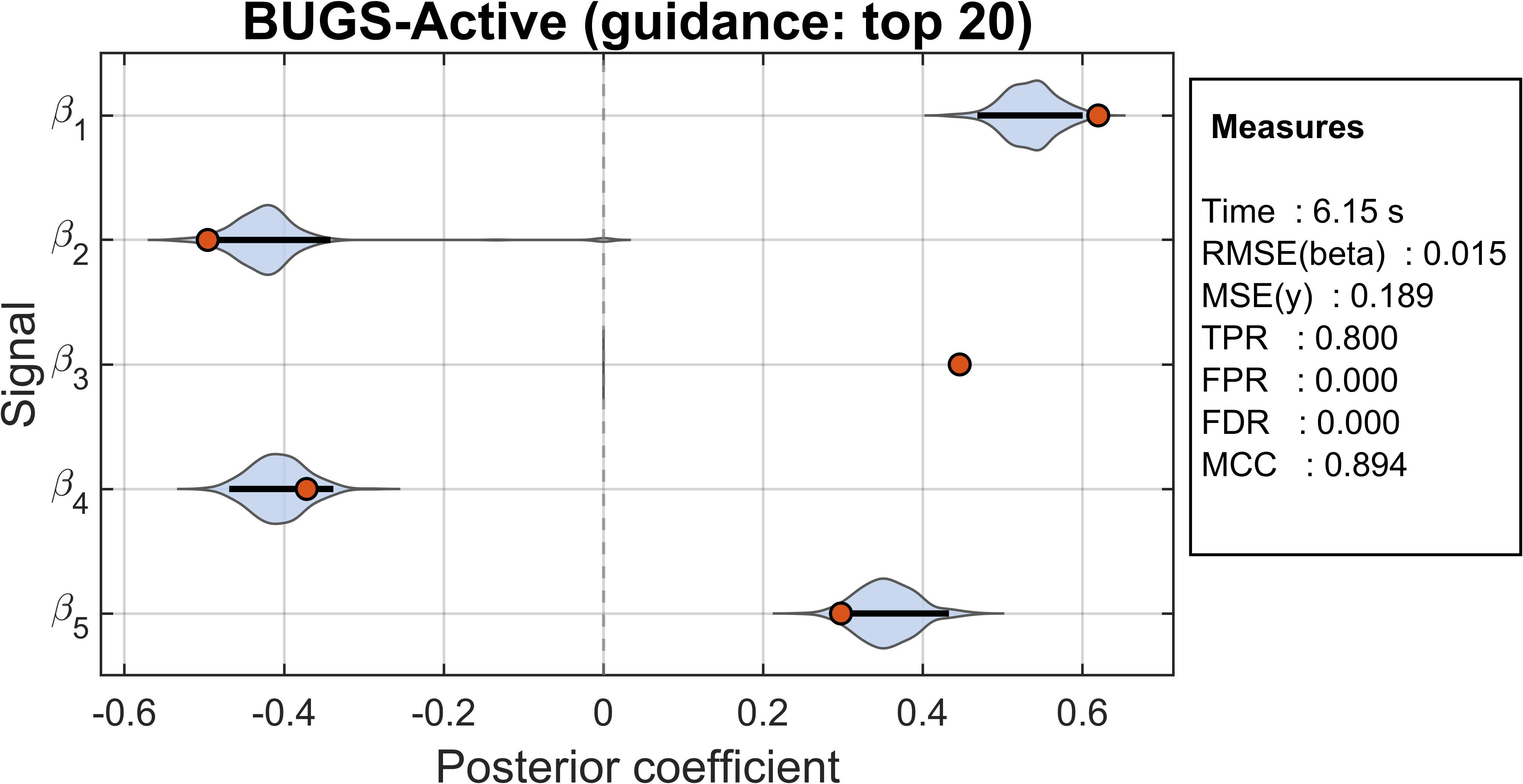}
\end{minipage}\hfill
\begin{minipage}[t]{0.495\linewidth}
\centering
\includegraphics[width=\linewidth]{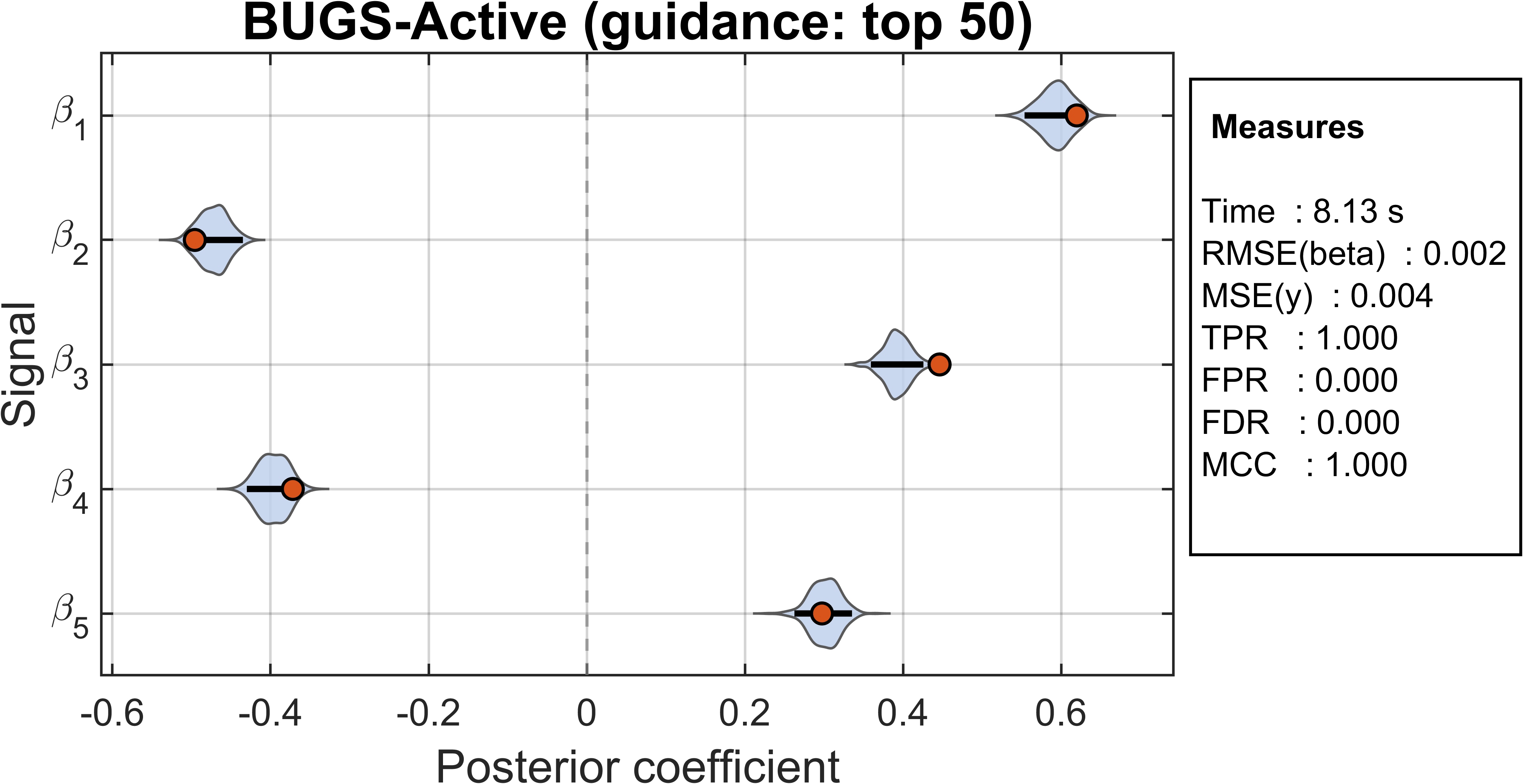}
\end{minipage}

\par\vspace{0.2cm}

% -------- Row 2 --------
\begin{minipage}[t]{0.495\linewidth}
\centering
\includegraphics[width=\linewidth]{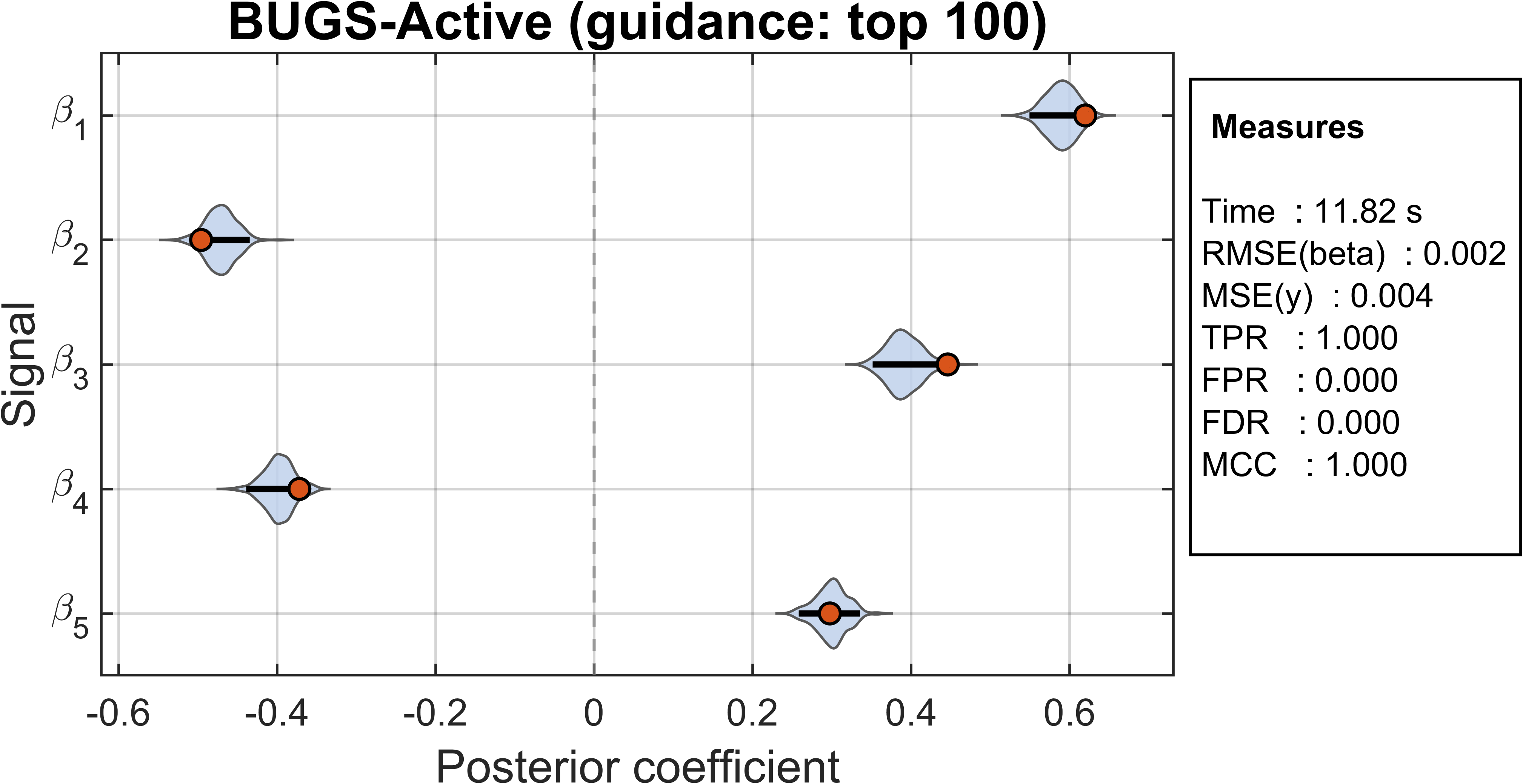}
\end{minipage}\hfill
\begin{minipage}[t]{0.495\linewidth}
\centering
\includegraphics[width=\linewidth]{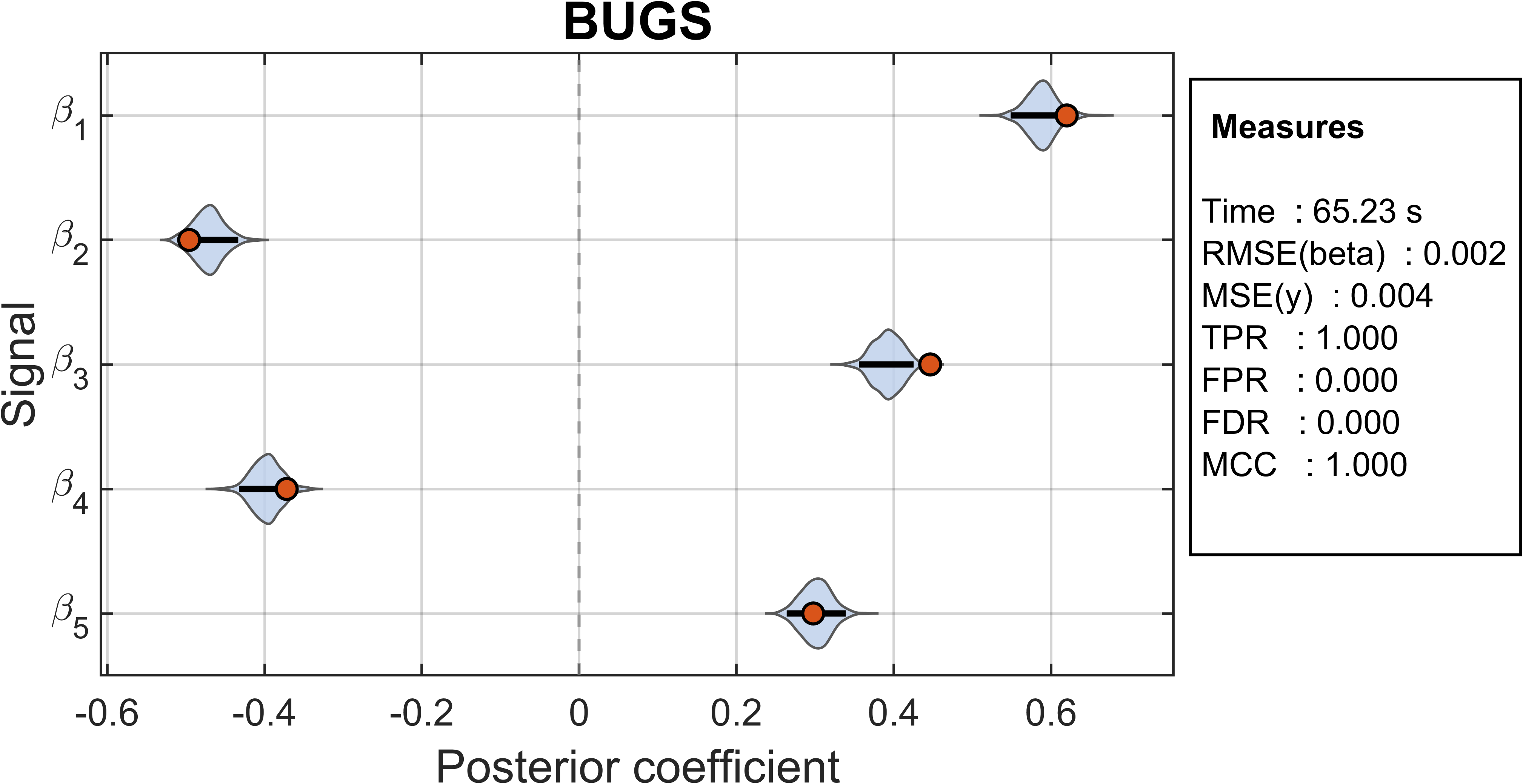}
\end{minipage}
\par

\caption{Effect of the guidance budget on posterior estimation in
BUGS-Active under Scenario~1 with $(p,n)=(1000,200)$ and $s_0=5$
(rather than $s_0=10$ in the original setting). Panels correspond
to guidance budgets of $20$ (top-left), $50$ (top-right), $100$
(bottom-left), and the full BUGS model without active-set restriction
(bottom-right). Each panel displays posterior summaries for the
top five true signals, including credible intervals and density
estimates, with true values indicated by red circles.}
\label{fig:guidance_effect}
\end{figure}
These observations also clarify the relationship between the practical BUGS-Active
algorithm and the screening conditions used in the theoretical analysis. The latter describe
an idealized regime in which the active set retains the true support while remaining small;
the practical algorithm approximates this behavior through a combination of marginal
guidance and posterior coefficient information. Consequently, the theoretical guarantees
should be interpreted as support for the active-set principle rather than as an exact
characterization of every finite-sample realization of the adaptive MCMC algorithm.
Developing data-adaptive rules for calibrating the guidance budget and other active-set
parameters is therefore an important direction for future work. The dimensional scope of the posterior contraction guarantees also warrants clarification. For column-standardized designs, the global operator-norm condition in Assumption (A11) implies $p=O(n)$, so the present contraction theory does not cover regimes in which $p/n\to\infty$. The ultra-high-dimensional simulations and application therefore provide empirical evidence of practical performance beyond the regime covered by these guarantees. Extending the contraction analysis to such regimes, under suitable design conditions, remains an important direction for future theoretical work.

A further limitation concerns the data-dependent construction of the guidance statistics.
The posterior contraction results assume deterministic guidance
or guidance constructed from auxiliary data independent of the
likelihood sample, whereas the default BUGS implementation
computes guidance from the same sample used in the likelihood. The sensitivity analysis in Section~\ref{sec:sim_study} showed high agreement
between same-sample and independent-guidance implementations in the settings considered,
with small mean absolute differences in TPR and FDR, providing empirical reassurance that
the practical procedure is not highly sensitive to this distinction in those experiments.
This comparison is nevertheless an empirical sensitivity assessment rather than a general
theoretical guarantee. A fuller treatment of the joint dependence between the guidance
statistics and posterior inference would strengthen the theoretical foundation of the method.

Finally, the present development focuses on Gaussian linear regression. Many scientific
applications involve binary, count, time-to-event, or longitudinal outcomes, for which both
the form of marginal guidance and posterior computation require further development.
Extensions to generalized linear and survival models, as well as guidance based on external
biological annotations or other auxiliary information, provide natural directions for future
work.

Taken together, the results suggest that marginal association information can be used
productively within Bayesian global--local shrinkage without requiring hard preliminary
screening. BUGS uses this information to sharpen sparse discovery, while BUGS-Active
extends the same principle to predictor dimensions approaching one million. Across the
simulation studies and the GUSTO methylation analysis, the resulting procedures retained
strong signal recovery and predictive performance while substantially reducing empirical false
discoveries in many settings. More broadly, the framework illustrates how inexpensive
marginal information can be incorporated as soft guidance within joint Bayesian inference,
rather than used solely as a dimension-reduction device.

\vspace{-0.5cm}
\section*{Acknowledgments}
None.

\vspace{-0.5cm}
\section*{Funding}
PD is partially supported by NIH-NCI Cancer Center Support Grant P30 CA016059.

\begin{center}
{\large\bf SUPPLEMENTARY MATERIAL}
\end{center}
\begin{description}
\item Codes for BUGS and BUGS-Active are available on GitHub at \href{https://github.com/priyamdas2/BUGS}{\pdl{https://github.com/priyamdas2/BUGS}}. Additional supplementary material provides details of posterior computation, full assumption statements and proofs, and supplementary
simulation, sensitivity, and convergence analyses.
\end{description}

\bibliographystyle{plainnat}
\bibliography{paper-ref}

@article{BuhlmannKalischMeier2014,
  author  = {B{\"u}hlmann, P. and Kalisch, M. and Meier, L.},
  title   = {High-Dimensional Statistics with a View Toward Applications in Biology},
  journal = {Annual Review of Statistics and Its Application},
  volume  = {1},
  pages   = {255--278},
  year    = {2014}
}

@article{FanSamworthWu2009,
  author  = {Fan, J. and Samworth, R. and Wu, Y.},
  title   = {Ultrahigh Dimensional Feature Selection: Beyond the Linear Model},
  journal = {Journal of Machine Learning Research},
  volume  = {10},
  pages   = {2013--2038},
  year    = {2009}
}

@article{WassermanRoeder2009,
  author  = {Wasserman, L. and Roeder, K.},
  title   = {High-Dimensional Variable Selection},
  journal = {The Annals of Statistics},
  volume  = {37},
  number  = {5A},
  pages   = {2178--2201},
  year    = {2009}
}

@article{Pidsley2016,
  author  = {Pidsley, R. and others},
  title   = {Critical Evaluation of the Illumina MethylationEPIC BeadChip Microarray for Whole-Genome DNA Methylation Profiling},
  journal = {Genome Biology},
  volume  = {17},
  pages   = {208},
  year    = {2016}
}

@article{Leroy2025,
  author  = {Leroy, A. and others},
  title   = {Longitudinal prediction of DNA methylation to forecast epigenetic outcomes},
  journal = {eBioMedicine},
  year    = {2025},
  volume  = {115},
  pages   = {105709}
}

@article{neal2003slice,
  title={Slice sampling},
  author={Neal, R.},
  journal={Annals of Statistics},
  volume={31},
  number={3},
  pages={705--767},
  year={2003}
}

@article{robbins1956empirical,
  title   = {An Empirical Bayes Approach to Statistics},
  author  = {Robbins, H.},
  journal = {Proceedings of the Third Berkeley Symposium on Mathematical Statistics and Probability},
  volume  = {1},
  pages   = {157--163},
  year    = {1956},
  publisher = {University of California Press}
}

@book{efron2010large,
  author    = {Efron, B.},
  title     = {Large-Scale Inference: Empirical Bayes Methods for Estimation, Testing, and Prediction},
  publisher = {Cambridge University Press},
  year      = {2010},
  address   = {Cambridge}
}

@article{fan2008sis,
  author  = {Fan, J. and Lv, J.},
  title   = {Sure Independence Screening for Ultrahigh Dimensional Feature Space},
  journal = {Journal of the Royal Statistical Society: Series B},
  volume  = {70},
  number  = {5},
  pages   = {849--911},
  year    = {2008}
}

@article{friedman2010glmnet,
  author  = {Friedman, Jerome and Hastie, Trevor and Tibshirani, Robert},
  title   = {Regularization Paths for Generalized Linear Models via Coordinate Descent},
  journal = {Journal of Statistical Software},
  year    = {2010},
  volume  = {33},
  number  = {1},
  pages   = {1--22}
}

@misc{bayesreg_matlab,
  author       = {Makalic, Enes and Schmidt, Daniel F.},
  title        = {{BayesReg: Flexible Bayesian Penalized Regression Modelling}},
  year         = {2020},
  howpublished = {\url{https://www.mathworks.com/matlabcentral/fileexchange/60823-flexible-bayesian-penalized-regression-modelling}}
}

@Manual{kang2025mhorseshoe,
    title = {Mhorseshoe: Approximate Algorithm for Horseshoe Prior},
    author = {Kang, M. and Lee, K.},
    year = {2025},
    note = {R package version 0.1.5}
  }

@article{bhattacharya2016fast,
  author  = {Bhattacharya, A. and Chakraborty, A. and Mallick, B.},
  title   = {Fast Sampling with Gaussian Scale Mixture Priors in High-Dimensional Regression},
  journal = {Biometrika},
  year    = {2016},
  volume  = {103},
  number  = {4},
  pages   = {985--991},
  doi     = {10.1093/biomet/asw042}
}

@article{bhadra2017horseshoeplus,
  author  = {Bhadra, A. and others},
  title   = {The Horseshoe+ Estimator of Ultra-Sparse Signals},
  journal = {Bayesian Analysis},
  year    = {2017},
  volume  = {12},
  number  = {4},
  pages   = {1105--1131}
}

@article{johndrow2020scalable,
  author  = {Johndrow, J. and Orenstein, P. and Bhattacharya, A.},
  title   = {Scalable Approximate MCMC Algorithms for the Horseshoe Prior},
  journal = {Journal of Machine Learning Research},
  year    = {2020},
  volume  = {21},
  number  = {73},
  pages   = {1--61}
}

@article{rockova2018ssl,
  author  = {Rockova, V. and George, E.},
  title   = {The Spike-and-Slab {LASSO}},
  journal = {Journal of the American Statistical Association},
  year    = {2018},
  volume  = {113},
  number  = {521},
  pages   = {431--444},
  doi     = {10.1080/01621459.2016.1260469}
}

@article{chatterjee2025unilasso,
  author  = {Chatterjee, S. and Hastie, T. and Tibshirani, R.},
  title   = {Univariate-Guided Sparse Regression},
  journal = {Harvard Data Science Review},
  year    = {2025},
  volume  = {7},
  number  = {3}
}

@article{bhattacharya2015dl,
  author  = {Bhattacharya, A. and others},
  title   = {Dirichlet--Laplace Priors for Optimal Shrinkage},
  journal = {Journal of the American Statistical Association},
  year    = {2015},
  volume  = {110},
  number  = {512},
  pages   = {1479--1490}
}

@article{zhang2022r2d2,
  author  = {Zhang, Y. and others},
  title   = {Bayesian Regression Using a Prior on the Model Fit: The {R2-D2} Shrinkage Prior},
  journal = {Journal of the American Statistical Association},
  year    = {2022},
  volume  = {117},
  number  = {538},
  pages   = {862--874}
}

@incollection{PolsonScott2011,
  author = {Polson, N. and Scott, J.},
  title = {Shrink globally, act locally: Sparse Bayesian regularization and prediction},
  booktitle = {Bayesian Statistics 9},
  editor = {Bernardo, J. M. and Bayarri, M. J. and Berger, J. O. and Dawid, A. P. and Heckerman, D. and Smith, A. F. M. and West, M.},
  pages = {501--538},
  year = {2011},
  publisher = {Oxford University Press}
}

@article{carvalho2010horseshoe,
  author = {Carvalho, C. and Polson, N. and Scott, J.},
  title = {The horseshoe estimator for sparse signals},
  journal = {Biometrika},
  volume = {97},
  number = {2},
  pages = {465--480},
  year = {2010}
}

@article{vanderPasKleijnvdV2014,
  author  = {van der Pas, S. and Kleijn, B. and van der Vaart, A.},
  title   = {The Horseshoe Estimator: Posterior Concentration around Nearly Black Vectors},
  journal = {Electronic Journal of Statistics},
  year    = {2014},
  volume  = {8},
  number  = {2},
  pages   = {2585--2618}
}

@article{vanderPasSzabovdV2017,
  author  = {van der Pas, S.  and Szab{\'o}, B. and van der Vaart, A.},
  title   = {Adaptive Posterior Contraction Rates for the Horseshoe},
  journal = {Electronic Journal of Statistics},
  year    = {2017},
  volume  = {11},
  number  = {2},
  pages   = {3196--3225}
}

@article{PiironenVehtari2017,
  author  = {Piironen, Juho and Vehtari, Aki},
  title   = {Sparsity Information and Regularization in the Horseshoe and Other Shrinkage Priors},
  journal = {Electronic Journal of Statistics},
  year    = {2017},
  volume  = {11},
  number  = {2},
  pages   = {5018--5051}
}

@article{GhosalGhoshvdV2000,
  title={Convergence rates of posterior distributions},
  author={S. Ghosal and J.K. Ghosh and A. van der Vaart},
  journal={The Annals of Statistics},
  volume={28},
  number={2},
  pages={500--531},
  year={2000},
  doi={10.1214/aos/1016218228}
}

@article{Zou2006,
author = {H. Zou},
title = {The Adaptive Lasso and Its Oracle Properties},
journal = {Journal of the American Statistical Association},
volume = {101},
number = {476},
pages = {1418--1429},
year = {2006},
doi = {10.1198/016214506000000735}
}

@article{park2008blasso,
  author = {T. Park and G. Casella},
  title = {The Bayesian Lasso},
  journal = {Journal of the American Statistical Association},
  year = {2008},
  volume = {103},
  number = {482},
  pages = {681--686},
  doi = {10.1198/016214508000000337}
}

@article{Tibshirani1996,
  author = {R. Tibshirani},
  title = {Regression Shrinkage and Selection via the {L}asso},
  journal = {Journal of the Royal Statistical Society: Series B},
  year = {1996},
  volume = {58},
  number = {1},
  pages = {267--288}
}

@article{Zou2005,
  author = {H. Zou and T. Hastie},
  title = {Regularization and Variable Selection via the Elastic Net},
  journal = {Journal of the Royal Statistical Society: Series B},
  year = {2005},
  volume = {67},
  number = {2},
  pages = {301--320},
  doi = {10.1111/j.1467-9868.2005.00503.x}
}

@article{Shokralla2012,
  author = {Shokralla, S. and others},
  title = {Next-generation sequencing technologies for environmental DNA research},
  journal = {Molecular Ecology},
  year = {2012},
  volume = {21},
  number = {8},
  pages = {1794--1805},
  doi = {10.1111/j.1365-294X.2012.05538.x}
}

@article{Caporaso2011,
  author = {Caporaso, J. and others},
  title = {Global patterns of 16S rRNA diversity at a depth of millions of sequences per sample},
  journal = {Proc. Natl. Acad. Sci.},
  year = {2011},
  volume = {108},
  number = {1},
  pages = {4516--4522},
  doi = {10.1073/pnas.1000080107}
}

@article{li_microbiome,
  author = {Hongzhe Li},
  title = {Microbiome, Metagenomics, and High-Dimensional Compositional Data Analysis},
  journal = {Annual Review of Statistics and Its Application},
  year = {2015},
  volume = {2},
  pages = {73--94},
  doi = {10.1146/annurev-statistics-010814-020351}
}

@Book{graphicx,
  author	= {Goossens, M. and Mittelbach, F. and Rahtz, A. and Roegel, D. and Vo\ss, H.},
  title 	= {The \LaTeX{} Graphics Companion},
  publisher	= {Addison--Wesley},
  address	= {Reading, MA},
  edition	= {Second},
  year  	= {2008}
}
\end{document}